\documentclass{sig-alternate}

\usepackage{hyperref}
\usepackage[caption=false,font=footnotesize]{subfig}
\usepackage{microtype}
\usepackage{graphicx}
\usepackage{array}
\usepackage{multirow}
\usepackage{bold-extra}
\usepackage{balance}

\usepackage[ruled]{algorithm2e}
\renewcommand{\algorithmcfname}{ALGORITHM}
\SetAlFnt{\small}
\SetAlCapFnt{\small}
\SetAlCapNameFnt{\small}
\SetAlCapHSkip{0pt}
\IncMargin{-\parindent}

\newtheorem{proposition}{Proposition}
\newtheorem{theorem}{Theorem}

\newtheorem{lemma}{Lemma}
\newdef{definition}{Definition}

\newcommand{\superscript}[1]{\ensuremath{^{\textrm{#1}}}}
\def\sharedaffiliation{\end{tabular}\newline\begin{tabular}{c}}

\def\upmc{\superscript{\dag}}
\def\thales{\superscript{\ddag}}

\newcommand{\G}{\mathcal{G}}
\newcommand{\R}{\mathcal{R}}
\newcommand{\Rb}{\bar{\R}}
\newcommand{\T}{\mathcal{T}}
\newcommand{\J}{\mathcal{J}}
\newcommand{\Low}{\mathcal{L}}
\newcommand{\Lb}{\bar{\Low}}
\newcommand{\Upp}{\mathcal{U}}
\newcommand{\N}{\mathbb{N}}
\newcommand{\X}{\mathcal{X}}

\newcommand{\D}{\mathcal{D}}

\title{Temporal Reachability Graphs
}

\numberofauthors{4}

\author{
John Whitbeck\upmc\thales\\
\email{john.whitbeck@lip6.fr}
\alignauthor
Marcelo Dias de Amorim\upmc\\
\email{marcelo.amorim@lip6.fr}
\alignauthor
Vania Conan\thales\\
\email{vania.conan@fr.thalesgroup.com}
\and
\alignauthor
\alignauthor
Jean-Loup Guillaume\upmc\\
\email{\mbox{jean-loup.guillaume@lip6.fr}}
\alignauthor
\sharedaffiliation
  \begin{tabular}{ccc}
    \affaddr{{\upmc}UPMC Sorbonne Universit{\'e}s} & & \affaddr{{\thales}Thales Communications \& Security}\\
    \affaddr{4, place Jussieu}                     & & \affaddr{160 Boulevard de Valmy}\\
    \affaddr{75005 Paris}                          & & \affaddr{92700 Colombes} \\
    \affaddr{France}                               & & \affaddr{France} \\
  \end{tabular}
}

\begin{document}

\conferenceinfo{MobiCom'12,} {August 22--26, 2012, Istanbul, Turkey.}
\CopyrightYear{2012}
\crdata{978-1-4503-1159-5/12/08}
\clubpenalty=10000
\widowpenalty = 10000

\maketitle

\begin{abstract}
While a natural fit for modeling and understanding mobile networks, time-varying graphs remain poorly understood. Indeed, many of the usual concepts of static graphs have no obvious counterpart in time-varying ones. In this paper, we introduce the notion of \emph{temporal reachability graphs}. A \textit{($\tau,\delta$)-reachability graph} is a time-varying directed graph derived from an existing connectivity graph. An edge exists from one node to another in the reachability graph at time $t$ if there exists a journey (i.e., a spatiotemporal path) in the connectivity graph from the first node to the second, leaving after $t$, with a positive edge traversal time $\tau$, and arriving within a maximum delay $\delta$. We make three contributions. First, we develop the theoretical framework around temporal reachability graphs. Second, we harness our theoretical findings to propose an algorithm for their efficient computation. Finally, we demonstrate the analytic power of the temporal reachability graph concept by applying it to synthetic and real-life datasets. On top of defining clear upper bounds on communication capabilities, reachability graphs highlight asymmetric communication opportunities and offloading potential.
\end{abstract}

\category{C.2.1}{Computer-Communication Networks}{Network Design and Architecture}[Store and forward networks]
\category{F.2.2}{Analysis of Algorithms and Problem Complexity}{Nonnumerical Algorithms and Problems}
\category{G.2.2}{Discrete Mathematics}{Graph Theory}

\terms{Algorithms, Theory, Performance}

\keywords{Time-varying graphs, Reachability, Opportunistic networks, Communication performance bounds}

\section{Introduction}

In time-varying graphs (TVG), vertices and edges appear and disappear as a function of time. Alternatively called temporal networks~\cite{Kempe2000} or evolving graphs~\cite{Xuan2003}, time-varying graphs have emerged over the past few years as a key model for a variety of complex systems~\cite{Dorogovtsev02,Harary97,temporal_graphs,Wagner01,Watts98}. Time-varying graphs can also serve as a solid theoretical framework for investigating fundamental properties of mobile networks~\cite{Casteigts2011,pellegrini07}. In Fig.~\ref{fig:tvg}, we show an example of a time-varying graph representing a five-node mobile network. 

\begin{figure}
  \centering
  \scalebox{1}{\includegraphics{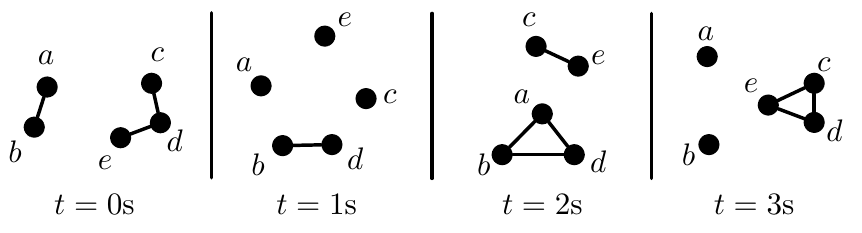}}
  \caption{Snapshots of a time-varying graph with five nodes.\label{fig:tvg}}
\end{figure}

Many of the typical concepts of static graphs such as paths, distance, diameter, or node degree have no obvious counterpart in time-varying graphs. Theorems that are true on static graphs may not hold in time-varying ones~\cite{Kempe2000} and dynamic equivalents of well-known problems on static graphs, such as finding strongly connected components, turn out to be intractable~\cite{Bhadra2002}. Naturally, it is always possible to track the evolution of static metrics such as node degree or clustering on snapshots of the time-varying graph at certain time intervals~\cite{Clauset2007} but, while instructive, this approach often fails to capture the correlations between successive snapshots~\cite{Holme2005}. Recently a number of concepts specific to time-varying graphs have been studied such as \emph{journeys}~\cite{Xuan2003}, \emph{temporal diameter}~\cite{chaintreau_diam}, or \emph{reachability time}~\cite{Holme2005}.

In the context of opportunistic mobile networking, much attention has been focused on aggregate inter-contact time distributions. These have been found to fit a power-law (with an $\alpha$ parameter smaller than 1) followed by an exponential cutoff in a number of real-life datasets~\cite{Chaintreau2007}. Should this power-law also hold for pairwise inter-contact times, then the expected values of these would be infinite, a pessimistic result indeed. However, further work has shown that aggregate power-law distributions can emerge from a diversity of pairwise inter-contact laws (including exponential)~\cite{Conan2007,Passarella2011}. Moreover, all information on multi-hop connectivity is ignored. Inter-contact time measurements do not therefore suffice to characterize connectivity in mobile networks.

\emph{Temporal reachability graphs} (TRG), our contribution, on the other hand, offer an immediate view of communication possibilities in a dynamic network. Given a time-varying graph $\mathcal{G}$, an arc exists from vertex $a$ to vertex $b$ at time $t$ in its derived \textit{($\tau,\delta$)-reachability graph} if a journey exists in $\mathcal{G}$ from $a$ to $b$ leaving $a$ after time $t$ and arriving at $b$ before $t+\delta$, given that each single-hop communication takes time $\tau$. Here, $\tau$ is the edge traversal time and $\delta$ is the upper bound on journey times in $\mathcal{G}$. By definition, temporal reachability are \emph{directed} time-varying graphs. For example, in the context of studying information dissemination in delay-tolerant networks (DTN)~\cite{dtn_fall_sigcomm}, $\tau$ would be the one-hop transmission delay and $\delta$ is the maximum tolerated delay.  

\begin{figure}
  \centering
  \scalebox{1}{\includegraphics{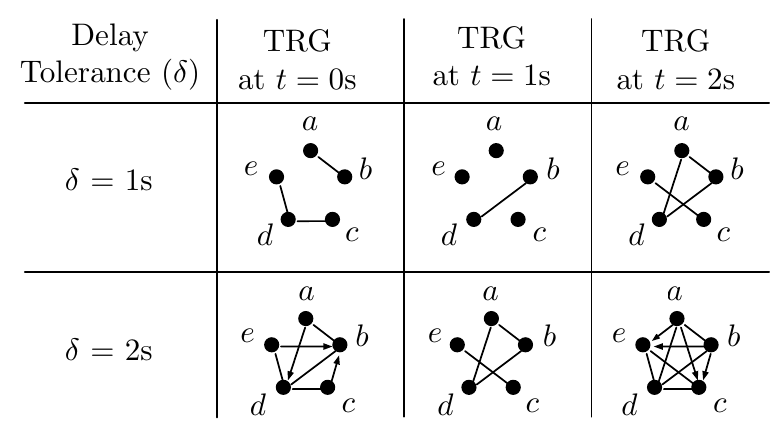}}
  \caption{Temporal reachability graphs obtained from the time-varying graph of Fig.~\ref{fig:tvg} for $\delta$=1s and 2s. Arrows indicate uni-directional communication opportunisties, while their absence indicates bi-directional ones. We show snapshots for different times $t$. We assume in this example that one-hop communications require $\tau$=1s.\label{fig:trg}}
\end{figure}

In Fig.~\ref{fig:trg}, we show the instances of the temporal-reachability graph extracted from the opportunistic mobile networks of Fig.~\ref{fig:tvg} at different times (assuming that the graph does not change between the snapshots shown in Fig.~\ref{fig:tvg}). We consider in this example that the link traversal time is $\tau$=1s. When $\delta=1$s, the only possible journeys in $\G$ at $t\in \{1s, 2s, 3s\}$ are the direct links between nodes~-- communications tolerate a delay of 1s, which is the time to traverse one single link.  When $\delta=2$s, several more journeys appear. Consider for instance the reachability graph at $t=0$s. There is a directed link between $a$ and $d$ as $a$ can wait one time unit for the link $(b,d)$ to appear ($a \leadsto b$ and $b \leadsto d$); thus, a two-hop journey $a \leadsto d$ through $b$ becomes possible. 

We make three main contributions:

\medskip\noindent\textbf{Temporal reachability graphs.} We formalize the concept of reachability graphs. In particular, under conditions that hold for all empirical datasets, we prove an additive property on the delay of reachability graphs. Roughly speaking, knowledge of the reachability graphs for delays $\delta$ and $\mu$ is enough to derive the reachability graph for delay $\delta+\mu$.

\medskip\noindent\textbf{Algorithm for efficient computation of TRG.} We translate this additive property into an efficient single-pass streaming algorithm~\cite{Alon1996} to calculate entire families of reachability graphs.

\medskip\noindent\textbf{Insights into real-world mobility traces.} Once calculated, reachability graphs yield many original insights on the temporal structure of the time-varying graphs they derive from. In this paper, we calculate the reachability graphs of several synthetic and real-life high-resolution connectivity traces~\cite{leboudec05,Musolesi2006,Salathe2010,Tournoux2011}. These highlight the concepts of \emph{temporal connectivity}, \emph{temporal asymmetry}, and \emph{temporal dominating set}. From the point of view of opportunistic communications, reachability graphs immediately provide, at all times, the maximum delivery ratio and the minimum set of users required to broadcast a message to the entire network, thereby offering a complete view of opportunistic communication capabilities and offloading potential.

\medskip The rest of the paper is structured as follows. Sections~\ref{sec:tvg} and~\ref{sec:temporal_reachability_graphs} introduce the concepts of time-varying and reachability graphs. Their theoretical properties justify their computation by a streaming algorithm detailed in Section~\ref{sec:algorithm}. Section~\ref{sec:results} then analyzes the reachability graphs of synthetic and real-life connectivity graphs. Related work is presented in Section~\ref{sec:related}. Finally, Section~\ref{sec:conclusion} concludes.

\section{Time-varying graphs}
\label{sec:tvg}

Time-varying graphs are a very useful high-level abstraction for studying connectivity over time in mobile network. Terms, notations, and definitions around time-varying graphs vary considerably, but recently, Casteigts et al. have proposed a unified framework for TVGs and, wherever applicable, we will use their definitions and notations~\cite{Casteigts2011}.

\begin{table}
  \centering
  \caption{Notations used in this paper}
  {\small \begin{tabular}{l|l}
    \textbf{Notation} & \textbf{Meaning} \\
    \hline \hline
    $N$ & Number of vertices \\ \hline
    $\G$ & A time-varying graph (TVG) \\ \hline
    $\tau$ & Edge traversal time of a TVG \\ \hline
    $\eta$ & Time step of an $\eta$-regular TVG \\ \hline
    \multirow{2}{*}{$\R_\delta$} & Reachability graph with maximum \\
    & delay $\delta$ derived from a TVG $\G$ \\ \hline
    $\| \R_\delta(t) \|$ & Number of arcs in $\R_\delta$ at time $t$ \\  \hline
    \multirow{2}{*}{$\Upp_\delta$, $\Low_\delta$}  & Upper and lower bounds of a reachability \\
    & graph $\R_\delta$. $\Low_\delta \subseteq \R_\delta \subseteq \Upp_\delta$. \\ \hline
    \multirow{2}{*}{$\D_\delta$} & A time-varying out-dominating set of a \\
    & reachability graph $\R_\delta$ \\ \hline
    \multirow{2}{*}{$\text{\textsf{sym}}(\R_\delta)$} & Symmetric subset of $\R_\delta$: $(u,v) \in$ \\
    & $\text{\textsf{sym}}(\R_\delta)(t) \Leftrightarrow (v,u) \in \text{\textsf{sym}}(\R_\delta)(t) $ \\ \hline
    \multirow{2}{*}{$\text{\textsf{asym}}(\R_\delta)$} & Asymmetric subset of $\R_\delta$: $(u,v) \in$ \\
    & $\text{\textsf{asym}}(\R_\delta)(t) \Leftrightarrow (v,u) \notin \text{\textsf{asym}}(\R_\delta)(t) $ \\
    \hline
  \end{tabular}}
  \label{table:notation}
\end{table}

\begin{definition} (\textbf{Time-varying graph})
\label{def:tvg}
Let $V$ be a set of vertices, and $E \subseteq V \times V$ the set of possible edges between vertices in $V$. Events occur over a time span $\T \subseteq \mathbb{T}$, where $\mathbb{T}$ is the temporal domain ($\N$ for discrete-time systems or $\mathbb{R}+$ for continuous-time ones). In the general case, a TVG is a tuple $\G = (V, E, \T, \rho, \zeta)$ where $\rho~: E \times \T \rightarrow \{0,1\}$, called \textit{presence function}, indicates whether a given edge exists at a given time and $\zeta~: E \times \T \rightarrow \mathbb{T}$, called \textit{latency function}, indicates the time it takes to cross a given edge if starting at a given date (the latency of an edge could vary in time).
\end{definition}

In this paper, we will mostly consider continuous-time TVGs (i.e., $\mathbb{T} = \mathbb{R}+$). Furthermore, we will always assume a \emph{constant} $\zeta$ function such that $\forall (e,t) \in E\times\T, \zeta(e,t) = \tau$ where $\tau \geq 0$ is our uniform edge traversal time. We also note $\G(t) \subseteq E$ the set of edges in the snapshot of $\G$ at time $t$. Hence $e \in \G(t) \Leftrightarrow \rho(e,t) = 1$. Finally, given an edge $e = (u,v) \in E$, we define $\text{\textsf{from}}(e) = u$ and $\text{\textsf{to}}(e) = v$.

\begin{definition} (\textbf{Inclusion})
\label{def:inclusion}
Let $\G$ and $\G'$ be two TVGs that differ only by their presence functions $\rho$ and $\rho'$. We write $\G \subseteq \G'$ if and only if $\forall t \in \T, \forall e \in E, \rho(e,t) \leq \rho'(e,t)$, or equivalently if and only if $\forall t \in \T, \G(t) \subseteq \G'(t)$.
\end{definition}

\begin{definition} (\textbf{Union})
\label{def:union}
Let $\G$, $\mathcal{H}$, and $\mathcal{I}$ be three TVGs that differ only by their presence functions. We write $\G = \mathcal{H} \cup \mathcal{I}$ if and only if $\forall t \in \T, \G(t) = \mathcal{H}(t) \cup \mathcal{I}(t)$.
\end{definition}

\begin{definition}(\textbf{Journey})
\label{def:journey}
A \emph{journey} in $\G$ is a sequence of couples $\J = \left\{ (e_1,t_1), (e_2,t_2), \dotsc, (e_k,t_k) \right\}$ such that $\forall i < k$: (i) $\text{\textsf{to}}(e_i) = \text{\textsf{from}}(e_{i+1})$; (ii) $\forall t\text{ s.t. }t_i \leq t < t_i + \tau, \rho(e_i,t) = 1$; and (iii) $t_{i+1} \geq t_i + \tau$. Here, $|\J| = k$ is this journey's \emph{topological length} (i.e., the number of hops). Furthermore, $\text{\textsf{departure}}(\J)$ and $\text{\textsf{arrival}}(\J)$ denote the starting date $t_1$ and the last date $t_k+\tau$, respectively. Finally $\delta_{\J} = \text{\textsf{arrival}}(\J) - \text{\textsf{departure}}(\J)$ is the journey's \emph{temporal length}. A journey may represent, for example, the sequence of hops that a message follows through an opportunistic network.
\end{definition}

\section{Temporal Reachability Graphs}
\label{sec:temporal_reachability_graphs}

Now that we have the necessary background, let us formally define the notion of \textit{temporal reachability graphs} (or simply \textit{reachability graphs} in the remainder of this paper, for the sake of readability). Furthermore, we write $\N^*$ for $\N \setminus \{0\}$. To the best of our knowledge, the definitions and results presented in this section are completely novel.

\subsection{Reachability Graphs}
\label{subsec:reachability_graphs}

\begin{definition}(\textbf{Reachability graph})
\label{def:reachability_graph}
For $\delta \in \mathbb{R}+$, let $\R_\delta$ be the reachability graph with maximum delay $\delta$ derived from $\G$ (with edge traversal time $\tau$). Formally, $\forall t$, $(u,v) \in \R_\delta(t)$ if and only if $u \neq v$ and there exists at time $t$ in $\G$ a journey $\J$ from $u$ to $v$ such that $\text{\textsf{departure}}(\J) \geq t$ and $\text{\textsf{arrival}}(\J) \leq t+\delta$.
\end{definition}

Note that if $\delta < \tau$ then $\R_\delta$ is empty, as even one-hop journeys do not have enough time to arrive before the maximum delay.

\begin{proposition}[\bfseries{\scshape{Growth}}]
\label{prop:growth}
Let $\R_\delta$ and $\R_\mu$ be two reachability graphs of $\G$. Then $\delta \leq \mu \implies \R_\delta \subseteq \R_\mu$.
\end{proposition}

This follows naturally from the definition of a reachability graph. Note that the reverse is not true. For example let $\G$ be a TVG with edge traversal time $\tau > 0$. Its derived reachability graphs $\R_{\tau/2}$ and $\R_{\tau/3}$ are both empty as their maximum delays are smaller than the time it takes to cross one edge. Therefore, as per definition~\ref{def:inclusion}, $\R_{\tau/2} \subseteq \R_{\tau/3}$ even though $\frac{\tau}{2} > \frac{\tau}{3}$.

\begin{definition}(\textbf{Composition})
\label{def:composition}
Let $\R_\delta$ and $\R_\mu$ be two reachability graphs of $\G$. We define their \emph{composition} $\R_\delta \otimes R_\mu$ as the TVG such that, at all times $t$,
\begin{displaymath}
(u,v) \in \left(\R_\delta \otimes \R_\mu\right)(t) \Leftrightarrow
\left\{
  \begin{array}{l}
    (u,v) \in \R_\delta(t)\text{, or} \\
    (u,v) \in \R_\mu(t+\delta)\text{, or} \\
    \exists w \in V, (u,w) \in \R_\delta(t) \text{ and } \\
    \quad (w,v) \in \R_\mu(t+\delta)
  \end{array} \right.
\end{displaymath}
\end{definition}

\begin{theorem}[\bfseries{\scshape{Decomposition}}]
\label{th:decomposition}
Let $\G$ be a TVG with edge traversal time $\tau$ and $\Rb = \{\R_\delta\}_{\delta \geq 0}$ the set of all its reachability graphs. Let $\delta \geq 0$ and $\mu \geq \tau$. Then:
\begin{displaymath}
\R_{\delta+\mu} = \bigcup_{0 \leq \epsilon < 1} \R_{\delta+\epsilon\tau} \otimes \R_{\mu-\epsilon\tau}.
\end{displaymath}
\end{theorem}

While not yet in a calculable form, this theorem is at the heart of our approach. Intuitively, it states that if all the reachability graphs with delays close to $\delta$ and $\mu$ are known, then one can calculate the reachability graph with delay $\delta+\mu$. A full formal proof is provided in the appendix (Section~\ref{subsec:decomp_proof}), but intuitively, it relies on the following simple idea. Any arc in $\R_{\delta+\mu}$ corresponds to a journey in $\G$ starting after time $t$ and arriving before $t+\delta+\mu$. Either it can be neatly divided into a (possibly empty) journey in $\R_{\delta}$ starting after $t$ and arriving before $t+\delta$ and another (possibly empty) journey in $\R_{\mu}$ starting after $t+\delta$ and arriving before $t+\delta+\mu$, or it starts crossing an edge before $t+\delta$ but after $t+\delta-\tau$ that straddles both time intervals. In the latter case, since the edge traversal time is $\tau$, one can incorporate the straddling edge into a journey arriving before $t+\delta+\epsilon\tau$ and the rest into another journey leaving after $t+\delta+\epsilon\tau$ of temporal length less than $\mu-\epsilon\tau$.

\subsection{Regular reachability graphs}
\label{subsec:regular_reachability}

Real-life datasets all have a \emph{maximum resolution} (e.g., a second or a millisecond) that corresponds to the precision with which they were measured. While it is tempting to map their time domain $\T$ to $\N$, in reality one must account for $0$-second edge durations. For example, if the edge traversal time $\tau$ is null, then even ephemeral edges that last $0$ seconds may be part of a journey. In the more common situation where $\tau$ is strictly positive, one may still encounter $0$-duration arcs in a reachability graph. For example, consider an edge $(u,v)$ that is present for only one second. If $\tau$ and $\delta$ are both also equal to one second, then a one-hop journey from $u$ to $v$ using that edge only exists at precisely the instant $t$ that the edge appears. This, in turn, corresponds to two $0$-second arcs in $\R_\delta(t)$: $(u,v)$ and $(v,u)$. Note that both of these arcs will become one-second-long arcs in the reachability graph when $\delta=2s$. These observations lead to the definition of \emph{regular time-varying graphs}.

\begin{definition}(\textbf{Regular TVG})
\label{def:regular}
A TVG $\G$ is an $\eta$-regular TVG if there exists $\eta > 0$ such that $\forall k \in \N, k\eta < t_1 \leq t_2 < (k+1)\eta \implies \G(t_1) = \G(t_2) \subseteq \G(k\eta)$. Here, $\eta$ is called $\G$'s resolution. The time interval $[k\eta,(k+1)\eta[$ is the $k^{th}$ epoch of $\G$ with starting time $k\eta$ and ending time $(k+1)\eta$.
\end{definition}

Without loss of generality, this definition assumes that the first epoch starts at $t=0$. Intuitively, an $\eta$-regular TVG is one whose instantaneous graph topology cannot change arbitrarily quickly, as, during each epoch, the graph topology remains constant. Ephemeral $0$-second edges or arcs may exist at the start of an epoch, but the TVG then remains constant until the start of the next epoch. Not only are regular TVGs a natural fit for real-life traces, but, as we will see, they also represent a class of TVGs whose reachability graphs are calculable. 

\begin{theorem}[\bfseries{\scshape{Regular reachability graphs}}]
\label{th:regular_reachability}
Let $\G$ be an $\eta$-regular TVG whose edge traversal time is $\tau \in \eta\N^*$. For $\delta \in \eta\N$, let $\R_\delta$ be a reachability graph of $\G$. Then $\R_\delta$ is an $\eta$-regular TVG and $\forall k \in \N, k\eta < t < (k+1)\eta \implies \R_\delta(t) \subseteq \R_\delta(k\eta) \cap \R_\delta\left((k+1)\eta\right)$. 
\end{theorem}

This theorem ensures that, as long as $\delta$ is a ``multiple'' of the resolution $\eta$, $\R_\delta$ carries on the $\eta$-regularity of $\G$. The full formal proof is a little technical (Section~\ref{subsec:proof_regular_reachability} in the appendix), but it relies on the following idea. Let us consider a one-hop journey in an $\eta$-regular TVG $\G$ leaving in the middle of an epoch at time $t$. Let $k$ be the integer such that $k\eta < t < (k+1)\eta$. It arrives in the next epoch at time $(k+1)\eta < t + \tau < (k+2)\eta$ (the theorem assumes that $\tau \geq \eta$). Because $\G$ is $\eta$-regular, the departure time of this one-journey can be nudged forwards or backwards as long as it remains with the same epoch. Hence all one-hop journeys with a departure time $t'$ such that $k\eta \leq t' \leq (k+1)\eta$ are also valid journeys in $\G$. This ``nudging'' can be extended, though not trivially, to multi-hop journeys which proves the theorem.

Furthermore, the $\tau\geq\eta$ hypothesis guarantees that an arc present in a reachability graph during an epoch is present not only at the start of the epoch like in any regular graphs both also at the start of the next epoch. This property will be leveraged later in this section. However $\R_\delta(t) \neq \R_\delta(k\eta) \cap \R_\delta\left((k+1)\eta\right)$. An example of such a situation is discussed in Section~\ref{subsec:approximation}.

\begin{theorem}[\bfseries{\scshape{Sampling}}]
\label{th:sampling}
Let $\G$ be an $\eta$-regular TVG whose edge traversal time is $\tau = n\eta$ with $n\in\N^*$. Let $\Rb = \{ \R_{i\eta} \}_{i\in\N}$ be the set of all its derived $\eta$-regular reachability graphs. For $(d,m) \in \N\times\N$ such that $m \geq n$, we have, $\forall a\in\N$:
\begin{displaymath}
\R_{(d+m)\eta}(a\eta) = \bigcup_{0 \leq k < n} \left( \R_{(d+k)\eta} \otimes \R_{(m-k)\eta} \right)(a\eta).
\end{displaymath}
\end{theorem}

This theorem is the adaptation of Theorem~\ref{th:decomposition} to regular TVGs. The formal proof is in the appendix (Section~\ref{subsec:proof_sampling}) but, as previously, it relies on the idea that within an epoch, there exists a little freedom to ``nudge'' journey departure and arrival times backwards or forwards. In a sense, if both the minimum departure time and maximal arrival time of a journey fall exactly at the start of a time epoch, then this journey can be divided into two sub-journeys whose departure and arrival times also map exactly to the start of epochs. Under this form, the value of a reachability graph at the start of each epoch can be \emph{exactly} calculated from the values, at the start of each epoch, of the proper set of $\tau$ pairs of reachability graphs.

Unfortunately, this exact formula does not extend to the time spent strictly within an epoch. In a way, the state of the reachability graph during epochs is more important than its state at their starting times. Indeed, the time spent at \emph{exactly} the start of epochs is infinitesimally small. Therefore any metric averaged over time (e.g. average density) will depend only on the states strictly within epochs. In the next section, we propose an upper and a lower bound on the reachability graph during epochs. As we will later see in Section~\ref{sec:results}, these upper and lower bounds are in fact nearly always equal and therefore achieve an excellent approximation of the real reachability graph.

\subsection{Upper and lower bounds}
\label{subsec:approximation}

We define for every regular reachability graph $\R_\delta$ an upper and lower bound, i.e., two TVGs $\Upp_\delta$ and $\Low_\delta$ such that $\Low_\delta \subseteq \R_\delta \subseteq \Upp_\delta$. Both of these bounds are equal to the reachability graph at the start of each epoch but only differ slightly during the epoch. The upper bound is a straightforward application of Theorem~\ref{th:regular_reachability}.

\begin{proposition}[\bfseries{\scshape{Upper Epoch approximation}}]
\label{prop:upper_epoch_approx}
Let $\G$ be an $\eta$-regular TVG whose edge traversal time is $\tau \in \eta\N^*$. For $\delta \in \eta\N$, let $\R_\delta$ be a reachability graph of $\G$. We define $\R_\delta$'s \emph{upper approximation}, the TVG $\Upp_\delta$, as follows ($\forall a\eta \leq t < (a+1)\eta$ with $a\in\N$):
\begin{itemize}
\item if $t=a\eta$, then $\Upp_\delta(a\eta) = \R_\delta(a\eta)$;
\item if $t > a\eta$, then $\Upp_{\delta}(t) = \R_\delta(a\eta) \cap \R_\delta\left((a+1)\eta\right)$.
\end{itemize}
As defined, $\forall \delta, \R_\delta \subseteq \Upp_\delta$.
\end{proposition}

The lower bound is a little more complicated and involves a modified version of the composition operator (Definition~\ref{def:composition}).

\begin{definition}(\textbf{Approximate composition})
\label{def:approximate_composition}
Let $\R_\delta$ and $\R_\mu$ be two reachability graphs of an $\eta$-regular TVG $\G$ with $\delta$ and $\mu$ in $\eta\N^*$. Let $\Low_\delta$ and $\Low_\mu$ be two TVGs such that $\Low_\delta \subseteq \R_\delta$ and $\Low_\mu \subseteq \R_\mu$. For all times $a\eta \leq t < (a+1)\eta$ with $a\in\N$, we define their \emph{approximate composition} $\Low_\delta \odot \Low_\mu$ as follows:
\begin{itemize}
\item if $t=a\eta$, then $\left( \Low_\delta \odot \Low_\mu \right)(a\eta) = \left( \R_\delta \otimes \R_\mu \right)(a\eta)$;
\item if $t>a\eta$, then
\end{itemize}
\begin{displaymath}
(u,v) \in \left(\Low_\delta \odot \Low_\mu\right)(t) \Leftrightarrow
\left\{
  \begin{array}{l}
    (u,v) \in \R_\delta\left((a+1)\eta\right)\text{, or} \\
    (u,v) \in \R_\mu(a\eta+\delta)\text{, or} \\
    \exists w \in V, (u,w) \in \Low_\delta(t) \text{ and } \\
    \quad (w,v) \in \Low_\mu(t+\delta)
  \end{array} \right.
\end{displaymath}
\end{definition}

By definition, at the start of each epoch the approximate composition is equal to the regular composition of the reachability graphs. During an epoch, the condition $(u,v) \in \R_\delta\left((a+1)\eta\right)$ is \emph{easier} to meet than the more intuitive $(u,v) \in \Low_\delta(t)$. Indeed, with Theorem~\ref{th:regular_reachability}, $\Low_\delta(t) \subseteq \R_\delta(t) \subseteq \R_\delta\left( (a+1)\eta \right)$, and this composition operator will catch the arcs (including the $0$-second ones) occurring at the start of the next epoch. The same is true for the $(u,v) \in \R_\mu(a\eta+\delta)$ condition for arcs ending at the start of the current ($\delta$-shifted) epoch. Combined, these more inclusive conditions make for a \emph{tighter} lower bound.

\begin{proposition}[\bfseries{\scshape{Lower Epoch approximation}}]
\label{prop:lower_epoch_approx}
Let $\G$ be an $\eta$-regular TVG whose edge traversal time is $\tau =n\eta$ with $n\in\N^*$. Let $\Rb = \{ \R_\delta \}_{\delta\in\eta\N}$ be the set of all its derived reachability graphs. We recursively define $\R_\delta$'s \emph{lower approximation} (i.e., the TVG $\Low_\delta$) as follows. For $d < 2n$, $\Low_{d\eta} = \R_{d\eta}$ and for all $(d,m)$ in $\N\times\N$ such that $d \geq n$ and $m \geq n$,
\begin{displaymath}
\Low_{(d+m)\eta} = \bigcup_{0 \leq k < n} \Low_{(d+k)\eta} \odot \Low_{(m-k)\eta}.
\end{displaymath}
As defined, $\forall d \in \N, \Low_{d\eta} \subseteq \R_{d\eta}$.
\end{proposition}

The proof, detailed in Section~\ref{subsec:proof_lower_epoch_approx} of the appendix, is fairly straightforward and involves setting the value of $\epsilon$ in Theorem~\ref{th:decomposition} as a function of $\eta$, $t$, $k$, and $n$.

This recursive definition means that the exact value of $\Low_\delta$ during a time epoch \emph{depends on the sequence of compositions that was used to calculate it.} For example, in the simple situation where $\tau=\eta=1$, the lower bounds $\Low_4 = \Low_2 \odot \Low_2$ and $\Low_4' = \Low_3 \odot \Low_1$ may be different. Proposition~\ref{prop:lower_epoch_approx} only guarantees that they are both included in the real reachability graph $\R_4$. However, as we will see, this is not a concern because this lower bound is tight regardless of how it is calculated.

\begin{figure}[t]
  \centering
  \includegraphics{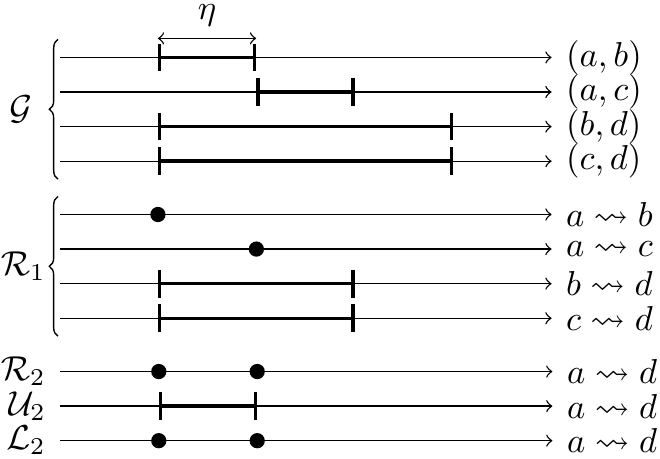}
  \caption{Upper epoch approximation with $\tau=\eta=1$. Arrows point in the direction of time. Intervals denote the presence of an edge or arc during that time. Dots are intervals reduced to a single point in time. In this example the lower bound $\Low_2$ is equal to the reachability graph $\R_2$ but the upper bound $\Upp_2$ is not.\label{fig:upper_approx}}
\end{figure}

Fig.~\ref{fig:upper_approx} is an example of a situation with four vertices where the upper bound has an arc during an epoch that is not in the reachability graph. Here $\tau=\eta=1$s. $\R_1$ is exactly derived from $\G$ by shortening the end time of each edge by $\tau$. For example, this leads to an ephemeral arc from $a$ to $b$ and a two-second-long arc from $b$ to $d$ (down from a three-second-long edge $(b,d)$ in $\G$). Obviously, not all arcs are represented on Fig.~\ref{fig:upper_approx}. In $\R_2$, there are two ephemeral arcs from $a$ to $d$, corresponding to the $a \leadsto b \leadsto d$ journey, and the $a \leadsto c \leadsto d$ journey an epoch later. In between those instants, it is both too late to use the first journey and too early to use the second. Since this arc exists at the start of two successive epochs, $\Upp_2$ errs in considering that it exists during the entire epoch. In this example, the lower bound $\Low_2$ matches the reachability graph.

\begin{figure}[t]
  \centering
  \includegraphics{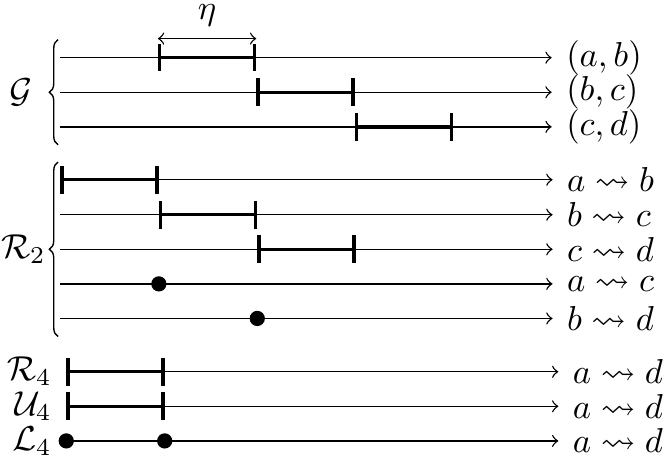}
  \caption{Lower epoch approximation with $\tau=\eta=1$. Arrows point in the direction of time. Intervals denote the presence of an edge or arc during that time. Dots are intervals reduced to a single point in time. In this example the upper bound $\Upp_4$ is equal to the reachability graph $\R_4$ but the lower bound $\Low_4$ is not.\label{fig:lower_approx}}
\end{figure}

Conversely, Fig.~\ref{fig:lower_approx} is an example of a situation with three vertices where the lower bound misses an arc during an epoch that is present in the reachability graph. Here again $\tau=\eta=1$. The one-second-long edges in $\G$ all become ephemeral $0$-second arcs in $\R_1$ (not shown on the figure), that in turn become one-second-long arcs in $\R_2$. Furthermore, $\R_2$ also contains two ephemeral arcs $(a,c)$ and $(b,d)$ that correspond to the $a \leadsto b \leadsto c$ and $b \leadsto c \leadsto d$ journeys respectively. Here $\R_2 = \Low_2 = \Upp_2$. Finally, $\R_4$ contains an arc $(a,d)$ that corresponds to the existence of the three-hop journey $ a \leadsto b \leadsto c \leadsto d$. $\Low_4$ correctly identifies this journey at two instants $k$ and $k+1$. At time $k$, it composes arc $(a,b)$ from $\R_2(k)$ with the ephemeral arc $(b,d)$ in $\R_2(k+2)$. At time $k+1$ it composes the ephemeral arc $(a,c)$ in $\R_2(k+1)$ with the arc $(c,d)$ in $\R_2(k+3)$. However, for $ k < t < k+1$, $\Low_4$ has no way of finding the journey from $a$ to $b$. Note that in this case, $\Upp_4$ is equal to $\R_4$.

The example of Fig.~\ref{fig:lower_approx} also helps us understand the absence of error propagation when calculating successive lower approximations from previous lower approximations. In this example, thanks to the conditions in Definition~\ref{def:approximate_composition} that refer to the starting times of epochs,  $\Low_8$ will ``bridge'' over the missing $(a,d)$ arc in $\Low_4$. However it may miss composed arcs that use the $(a,d)$ arc in $\R_4$. These missed arcs will be, in-turn, ``bridged'' in, for instance, a $\Low_{16}$ approximation. Roughly speaking, errors during epochs do not propagate beyond two compositions. Combined with the fact that the lower approximation matches the reachability graph at the start of each epoch, this leads, as we will verify empirically, to a tight lower approximation or the reachability graph.

\subsection{A few case studies}
\label{subsec:simple}

In this section, we examine how the theoretical results above apply to some simple situations.

\medskip\noindent\textbf{Zero edge traversal time.}
The reachability graphs of $\eta$-regular TVGs with zero edge traversal time ($\tau=0$) are very easy to calculate. In particular $(u,v) \in \R_0(t)$ if and only if $u$ and $v$ are in the same connected component at time $t$. Similarly, $\R_\eta(t)$ can be calculated from the connected components at time $t$ and $t+\eta$. Furthermore, Theorem~\ref{th:decomposition} becomes: $\R_{\delta+\mu} = \R_\delta \otimes \R_\mu$. This relation can then be applied repeatedly (e.g., in a binary-exponentiation) to obtain any $\R_\delta$ with $\delta \in \eta\N^*$. Furthermore, the reachability graphs thus calculated are \emph{exact}. However, in this specific case where the edge traversal time is null, more efficient algorithms for calculating reachability graphs exist. For example the algorithm proposed by Chaintreau et al. for computing delay-optimal paths could be easily adapted to this purpose~\cite{chaintreau_diam}.

\medskip\noindent\textbf{Several hops per epoch.}
While all of the results on regular reachability graphs in Sections~\ref{subsec:regular_reachability} and~\ref{subsec:approximation} assume that the edge traversal time $\tau$ is in $\eta\N^*$, they can be simply adapted to the case where $\tau = \eta / n$ with $n\in\N^*$, i.e., when up to $n$ edges may be crossed during a single epoch. Indeed, in this case, any $\eta$-regular graph would also be $\tau$-regular, and we can apply all of our results as if $\tau=\eta$. For example, Proposition~\ref{prop:lower_epoch_approx} is reduced to $\Low_{(d+m)\eta} = \Low_{d\eta} \odot \Low_{m\eta}$.

\medskip\noindent\textbf{Unit delay.}
For $\tau>0$, the reachability graph $\R_\tau$ is trivially calculable from $\G$. Indeed, for any edge $(u,v)$ in $\G$ that appears at time $t_1$ and disappears at time $t_2 \geq t_1+\tau$, the arcs $(u,v)$ and $(v,u)$ appear in $\R_\tau$ at time $t_1$ and disappear at time $t_2-\tau$. This derivation is simple but essential for bootstrapping iterations of lower bound compositions (e.g., repeated applications of Proposition~\ref{prop:lower_epoch_approx}).

\section{Efficient computation of reachability graphs}
\label{sec:algorithm}

\subsection{Families of reachability graphs}
\label{subsec:families}

Another interesting property of lower approximations of reachability graphs is that an upper approximation may be derived from it. Indeed, since it contains the exact values of the reachability graph at the start of each epoch, it is trivial to calculate their intersection during each epoch. While the method described in this section focuses on the efficient computation of the lower approximations of reachability graphs, it simultaneously computes the upper approximations.

At a high level, the algorithm presented in this section is a \textit{binary exponentiation} on families of lower bounds of reachability graphs using a special additive operator. To simplify notations, we will consider in this section that $\eta=1$ and $\tau \in \N^*$.

\begin{definition}(\textbf{Lower bound family})
\label{def:family}
Let $\G$ be a 1-regular TVG whose edge traversal time is $\tau \in \N^*$. For $d \in \tau\N$, let $\Low_d$ be the lower approximation of a reachability graph $\R_d$ of $\G$. For $d \geq \tau$, we define $\Lb_d$, the \textit{family} of $\Low_d$ such that $\Lb_d = \left\{ \Low_{d+i} \right\}_{-\tau < i < \tau}$.
\end{definition}

\begin{proposition}[\bfseries{\scshape{Family additivity}}]
\label{prop:family_decomp}
Let $\G$ be a 1-regular TVG whose edge traversal time is $\tau \in \N^*$. Let $\R_d$ and $\R_m$ be two reachability graphs of $\G$ such that $d \in \tau\N^*$ and $m \in \tau\N^*$, and $\Lb_d$ and $\Lb_m$ the respective families of their lower approximations. We define $\Lb_d \oplus \Lb_m$ the set of TVGs such that
\begin{displaymath}
\begin{array}{ll}
\Lb_d \oplus \Lb_m = & \left\{ \bigcup_{0 \leq k < \tau} \Low_{d+k} \odot \Low_{m+i-k}\right\}_{0 \leq i < \tau} \\ 
& \bigcup \left\{ \bigcup_{0 \leq k < \tau} \Low_{d+i+k} \odot \Low_{m-k}\right\}_{-\tau < i < 0}
\end{array}
\end{displaymath}
Then (i) any element $\Low_{d+m+i} \in \Lb_d \oplus \Lb_m$ with $-\tau < i < \tau$ is a lower approximation of $\R_{d+m+i}$, and (ii) any TVG in $\Lb_{d+m}$ can be calculated from $\tau$ compositions of pairs of TVGs in $\Lb_d \times \Lb_m$.
\end{proposition}

\begin{proof}
Consider any lower approximation $\Low_{d+m+i} \in \Lb_{d+m}$. If $-\tau < i < 0$, $\Low_{d+m+i} = \Low_{(d+i)+m} =$ \\$\bigcup_{0 \leq k < \tau} \Low_{d+i+k} \odot \Low_{m-k}$ (Theorem~\ref{prop:lower_epoch_approx}). If $0 \leq i < \tau$, $\Low_{d+m+i} = \Low_{d+(m+i)} = \bigcup_{0 \leq k < \tau} \Low_{d+k} \odot \Low_{m+i-k}$ (Theorem~\ref{prop:lower_epoch_approx}).
\end{proof}

\begin{figure}
  \centering
  \includegraphics{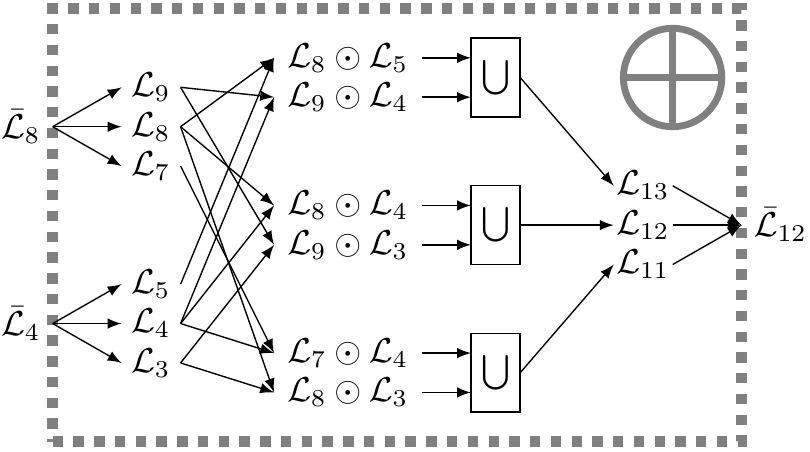}
  \caption{Adding $\Lb_4$ and $\Lb_8$ families to obtain an $\Lb_4 \oplus \Lb_8 = \Lb_{12}$ family. In this example $\tau=2$ and $\eta=1$.\label{fig:adding}}
\end{figure}

We now have our self-sufficient elements, the lower bound families, and an additive operation $\oplus$ between them. For illustration, Fig.~\ref{fig:adding} details the process of adding $\Lb_4$ and $\Lb_8$ families to obtain an $\Lb_4 \oplus \Lb_8 = \Lb_{12}$ family. The $\Low_3$, $\Low_4$, $\Low_5$, $\Low_7$, $\Low_8$, and $\Low_9$ lower bounds are combined as inputs for three applications of Proposition~\ref{prop:lower_epoch_approx} that yield $\Low_{11}$, $\Low_{12}$, and $\Low_{13}$. Viewed as a black box, this operation combines the $\Lb_4$ and $\Lb_8$ families into an $\Lb_{12}$ family.

The inner workings of the $\oplus$ operator are embarrassingly parallel. Indeed, each application of Proposition~\ref{prop:lower_epoch_approx} can be run completely independently of the others. This opens the way for highly distributed implementations, whose speed will be determined by that of the composition of $\tau$ pairs of lower bounds of reachability graphs. Accordingly, the next section proposes an efficient algorithm for this composition operation.

\subsection{Composing reachability graphs}
\label{subsec:composing}

A time-varying graph $\G$ may be stored as a time-indexed sequence of edge $\text{\textsf{UP}}$ and $\text{\textsf{DOWN}}$ events. For example, if at time $t$ an event $\{ (u,v), \text{\textsf{DOWN}} \}$ occurs, then the edge $(u,v)$ disappears at time $t$ in $\G(t)$. Such a representation is well suited for algorithms that sequentially examine all states of the TVG.

We present a streaming algorithm for composing $\tau$ reachability graph lower bounds as in Proposition~\ref{prop:lower_epoch_approx}. Streaming algorithms are well suited to TVGs as their memory requirements do not depend on the duration of the trace but only on the number of vertices~\cite{Alon1996}. Memory is indeed a limited resource, as reachability graphs can become fully connected cliques and data structures such as adjacency matrices cannot therefore be considered sparse. Furthermore, in our case, a streaming algorithm facilitates a parallel implementation of the $\oplus$ operator as one process can read the input families, duplicate their events, and dispatch these to various workers, each calculating one composition operation.

\newcommand{\pp}{++}
\newcommand{\mm}{--\kern0.2ex--}

\begin{algorithm}[t]
\DontPrintSemicolon
\SetAlgoVlined
\LinesNumbered
\SetKwInput{Require}{Require}
\SetKwInput{Ensure}{Ensure}
\SetKwInput{Data}{Local}
\Require{$\{ \Low_{d+k}, \Low_{m-k} \}$ \tcp{$\tau$ pairs of input streams} }
\Ensure{$\Low_{d+m}$ \tcp{the output stream} }
\Data{$\{ adj_{d+k}, adj_{m-k} \}$, \tcp{$\tau$ pairs of adjacency matrices}}
\Data{$counter[]$ \tcp{arc counter}}
\Data{$delayed_k[]$ \tcp{$\tau$ lists of delayed down events}}
\For{$i \leftarrow 0$ \KwTo $T$}{
  \For{$k \leftarrow 0$ \KwTo $\tau-1$}{
    \lForAll{$\{a,\text{\textsf{UP}}\} \in \Low_{d+k}[i]$}{$counter[a]$\pp} \label{alg:deltaup}
  }
  \ForAll{new arcs $a$}{ append $\{a,\text{\textsf{UP}}\}$ to $\Low_{d+m}[i-1]$  } \label{alg:up}
  \ForAll{arcs $a$ s.t. $counter[a]=0$}{ append $\{a,\text{\textsf{DOWN}}\}$ to $\Low_{d+m}[i-1]$  } \label{alg:down}
  \For{$k \leftarrow 0$ \KwTo $\tau-1$}{
    \ForAll{$\{(u,v),\text{\textsf{UP}}\} \in \Low_{d+k}[i]$}{
      add $(u,v)$ to $adj_{d+k}$ \;  \label{alg:duj1}
      \lForAll{$(v,w) \in adj_{m-k}$}{$counter[ (u,w) ]$\pp} \label{alg:duj2}
    }
    \ForAll{$\{(u,v),\text{\textsf{UP}}\} \in \Low_{m-k}[i+d+k]$}{ 
      $counter[ (u,v) ]$\pp \; \label{alg:muup1}
      add $(u,v)$ to $adj_{m-k}$ \; \label{alg:muup2}
      \lForAll{$(w,u) \in adj_{d+k}$}{$counter[ (w,v) ]$\pp} \label{alg:muup3}
    }
    \lForAll{$a \in delayed_k$}{$counter[k]$\mm} \; \label{alg:mudelayeddown}
    $delayed_k \leftarrow \emptyset$ \; \label{alg:cleardelayed}
    \ForAll{$\{(u,v),\text{\textsf{DOWN}}\} \in \Low_{m-k}[i+d+k]$}{ \label{alg:mudown}
      add $(u,v)$ to $delayed_k$ \; \label{alg:mudown1}
      remove $(u,v)$ from $adj_{m-k}$ \; \label{alg:mudown2}
      \lForAll{$(w,u) \in adj_{d+k}$}{ $counter[(w,v)]$\mm} \label{alg:mudown3}
    }
    \ForAll{$\{(u,v),\text{\textsf{DOWN}}\} \in \Low_{d+k}[i]$}{ \label{alg:deltadown1}
      $counter[ (u,v) ]$\mm \; \label{alg:deltadown2}
      remove $(u,v)$ from $adj_{d+k}$ \; \label{alg:deltadown3}
      \lForAll{$(v,w) \in adj_{m-k}$}{ $counter[(u,w)]$\mm} \label{alg:deltadown4}
    }
  }
}
\caption{Lower bound composition\label{alg:composition}}
\end{algorithm}

Algorithm~\ref{alg:composition} reads its input from $\tau$ pairs of event streams $\{\Low_{d+k}, \Low_{m-k} \}$ and writes $\Low_{d+m} = \bigcup_{0 \leq k < \tau} \Low_{d+k} \odot \Low_{m-k}$ to its output stream. It makes use of an arc counter that tracks how many of the conditions in the definitions of $\Low_{d+k} \odot \Low_{m-k}$ (see Definition~\ref{def:approximate_composition}) are verified by each arc. When a previously down link fulfills one of these conditions, its counter is initialized to $1$ and an $\text{\textsf{UP}}$ event is written to the output stream (Line~\ref{alg:up}). When this counter goes to $0$ the arc is removed from $\Low_{d+m}$ and a $\text{\textsf{DOWN}}$ event is written to the output stream (Line~\ref{alg:down}). Note that an arc may be brought up and down at the same time $i$ if its counter goes to $0$ right after it appears (i.e., an ephemeral arc). In this case, the arc triggers both lines~\ref{alg:up} and~\ref{alg:down}.

In more detail, this algorithm directly maps to the three conditions in the definition of the approximate composition $\odot$ (Definition~\ref{def:approximate_composition}):	

\medskip\noindent\textbf{Condition 1.} $(u,v) \in \R_{d+k}(i+1)$. This condition is handled on line~\ref{alg:deltaup}, \textit{before} making a decision on bringing arcs up or down at the \textit{previous} epoch. Indeed, the output stream at time $i$ can only be written to after reading the input streams up to time $i+1$.

\medskip\noindent\textbf{Condition 2.} $(u,v) \in \R_{m-k}(i+d+k)$. This condition means that down events in $\Low_{m-k}$ must be delayed for one epoch before lowering an arc's counter. This accounts for the local $delayed_k$ lists that are processed on lines~\ref{alg:mudelayeddown} to~\ref{alg:mudown1}.

\medskip\noindent\textbf{Condition 3.} $\exists w$, $(u,w) \in \Low_{d+k}(t)$ and $(w,v) \in \Low_{m-k}(t+d)$. This condition is checked by maintaining for each $k$ two adjacency lists $adj_{d+k}$ and $adj_{m-k}$ and checking upon $\text{\textsf{UP/DOWN}}$ events whether the end of an arc in $adj_{d+k}$ corresponds to the origin of an arc in $adj_{m-k}$ (lines~\ref{alg:duj1}-\ref{alg:duj2}, \ref{alg:muup2}-\ref{alg:muup3}, \ref{alg:mudown2}-\ref{alg:mudown3}, and~\ref{alg:deltadown3}-\ref{alg:deltadown4}).

\medskip Algorithm~\ref{alg:composition} is then used as a building block to implement the $\oplus$ addition of lower bound families. Starting from $\Low_\tau$'s family, for any $n \in \N^*$, $\Low_{n\tau}$ is obtained in $\log(n)$ applications of the $\oplus$ operation using a binary exponentiation process.

An implementation of Algorithm~\ref{alg:composition}, fully integrated into a binary exponentiation algorithm over lower bound families, is available as a part of our dynamic trace library (DiTL~\cite{DITL}). This package also contains the code for transforming a lower bound TVG into an upper bound TVG, as well as the time-varying dominating set computation used in Section~\ref{sec:results}.

\subsection{Complexity analysis}
\label{subsec:complexity analysis}

The worst-case memory requirements for this one-pass streaming algorithm are straightforward. Its local memory contains a non-sparse arc counter that requires up to $O(N^2)$ space. Furthermore, for each $0 \leq k < \tau$, it maintains two non-sparse adjacency matrices and an arc event list thereby requiring $O(N^2)$ space. Adding everything together yields a worst-case space complexity of $O(\tau N^2)$ that is \textit{independent} of the duration of the trace.

Before examining the worst-case time complexity, a word must be said about the implementation of the adjacency matrices (the arc counter is backed by an adjacency matrix). In our implementation, these are backed by per-vertex hash tables. Insertion and removal are therefore constant time operations but this approach may not scale to TVGs with much greater number of vertices than those considered in this paper. In this analysis we will consider that insertion and removal cost $O(\log(N))$ (e.g., by using binary trees).

At each epoch, arcs are brought up and/or down. For each $0 \leq k < \tau$, we note $M_{d+k}$ and $M_{m-k}$ the number of $\text{\textsf{UP}}$ and $\text{\textsf{DOWN}}$ events in the entire TVGs $\Low_{d+k}$ and $\Low_{m-k}$, respectively. Let $M = \max_k \left\{ \max\{M_{d+k},M_{m-k}\} \right\}$ be the maximum number of events in all the involved TVGs. In the worst case scenario where all possible arcs are updated at each epoch, we have $M = O(TN^2)$ events, where $T$ is the number of epochs. Each processed event costs a modification of the arc counter ($O(\log(N))$), an adjacency matrix ($O(\log(N)$), and up to $N$ counters for the composed paths ($O(N\log(N))$). All of these must be performed for the $\tau$ pairs of input TRGs. Adding everything together yields a worst-case time complexity of $O\left(\tau M N\log(N) \right)$, or $O\left(\tau T N^3\log(N) \right)$ in terms of $N$ exclusively. In practice however $M \ll TN^2$, so the first formulation is more accurate.

\section{Applications}
\label{sec:results}

Having formalized reachability graphs and detailed a method for efficiently computing upper and lower approximations of them in the previous sections, we now study the reachability graphs of several synthetic and real-life traces for a variety of edge traversal times $\tau$ and maximum delays $\delta$. In particular, we show how the properties of these reachability graphs place bounds on communication capabilities and highlight the asymmetric nature of dynamic networks.

\subsection{Datasets and metrics}
\label{subsec:datasets}

\begin{table}
  \centering
  \caption{Dataset characteristics. $N$ is the number of vertices, $M$ the number of edge UP/DOWN events, $T$ the duration, $P$ the beaconing period, and $\eta$ the time resolution.\label{table:datasets}\vspace*{2mm}}
  {\small \begin{tabular}{l|c|c|c|c|c}
    \textbf{Name} & $N$ & $M$ & $T$ & $P$ & $\eta$ \\
    \hline \hline
    Rollernet & $62$ & $99k$ & $3h$ & $15s$ & $1s$ \\
    Stanford & $782$ & $704k$ & $8h$ & $20s$ & $20s$ \\
    \hline
    Random Waypoint & $50$ & $46k$ & $8h$ & $1s$ & $1s$ \\
    Community & $50$ & $824k$ & $8h$ & $1s$ & $1s$ \\
    \hline
  \end{tabular}}
\end{table}

While we have calculated reachability graphs on many publicly available datasets, in this paper we present results based on two real-life contact traces selected for their short beaconing periods:

\medskip\noindent\textbf{Stanford~\cite{Salathe2010}.} As part of an epidemiology study, this trace captures face-to-face contacts among all students, teachers, and staff in a US high school between 7~a.m. and 4~p.m. The 782 ZigBee motes (TelosB Crossbow) sent beacons every 20 seconds (sending times are synchronized in the published trace).

\medskip\noindent\textbf{Rollernet~\cite{Tournoux2011}.} Opportunistic sighting of Bluetooth devices by groups of rollerbladers carrying Intel iMotes during a roller tour. The 62 iMotes performed neighborhood scans every 15 seconds. 

\medskip For comparison purposes, we also provide results based on two well understood synthetic mobility models:

\medskip\noindent\textbf{Random-Waypoint~\cite{leboudec05}.} We simulated $50$ nodes with speeds between 3 and 7~m/s from the stationary state in a $1,000\times 500$m$^2$ rectangle sending beacons every second with a 20~m transmission range.

\medskip\noindent\textbf{Community Model~\cite{musolesi05}.} We simulated $50$ nodes using the same parameters as above with 8 communities. In the community model, nodes with stronger social ties are more likely to be in geographic proximity.

\medskip The characteristics of these four $\eta$-regular traces are summarized in Table~\ref{table:datasets}.

For each reachability graph, we compute its time-varying dominating set (TVDS), i.e., a mutable set such that, at all times, there exists an incoming arc to any vertex in the TVG from a member of the TVDS. Thus, at any given time $t$, the TVDS is a traditional dominating set on the directed graph at time $t$. We note this TVDS $\D_\delta$ if it is derived from a reachability graph $\R_\delta$. 

Even on static graphs, calculating a minimal dominating set is a classic NP-complete problem. Here we adapt the well-known greedy algorithm for choosing multipoint relays for broadcasting in a wireless network to the time-varying context~\cite{laouiti2001}. The size of the dominating set calculated by this algorithm is within a factor $\log(N)$ of the optimal, where $N$ is the number of nodes. Our TVDS calculation algorithm reacts to arc UP/DOWN events as follows.
\begin{enumerate}
\item If, after the arc event is processed, the previous dominating set is no longer a dominating set, then proceed to step \ref{tvds:calc}. Otherwise do nothing.
\item \label{tvds:calc} Iteratively build the new dominating set, starting from an empty set, by greedily adding the node with the greatest outgoing degree to nodes not yet covered by the new dominating set. If two nodes have equal outgoing degrees, pick the one that is in the previous dominating set.
\end{enumerate}
The time-varying dominating set thus calculated attempts to be both reasonably stable over time and close the optimal.

\begin{table}
  \centering
  \caption{Metrics on $\eta$-regular reachability graphs of duration $T$. $\X_k$ is the value of $\X$ during the $k^{th}$ epoch, i.e., $\X_k = \X\left( (k+\frac{1}{2})\eta \right)$.\label{table:metrics}\vspace*{2mm}}
  {\small \begin{tabular}{m{2.6cm}|m{5cm}}
    \textbf{Metric name} & \textbf{Definition} \\
   \hline \hline
    Avg. dominating set size & $\frac{\eta}{TN} \sum_{k=0}^{T/\eta} \| \D_k \|$ \\ \hline
    Avg. density & $\frac{\eta}{TN(N-1)} \sum_{k=0}^{T/\eta} \| \R_k \|$ \\ \hline
    Avg. asymmetry & $\frac{\eta}{T} \sum_{k=0}^{T/\eta} \frac{\| asym(\R_k) \| }{\| asym(\R_k) \| + \| sym(\R_k) \| /2 } $ \\
    \hline
  \end{tabular}}
\end{table}

The rest of this section focuses on the following metrics. Their formulas and average values are found in Table~\ref{table:metrics}.

\begin{itemize}

	\item \textbf{Density.} The ratio of the number of arcs in a reachability graph at time $t$ over the total number of possible arcs.

	\item \textbf{Dominating set size}. The normalized number of vertices in the dominating set at time $t$.

	\item \textbf{Asymmetry}. Here we no longer consider directed arcs but undirected pairs of vertices. At time $t$, the asymmetry is the ratio of asymmetric pairs among pairs that have at least one edge between them.

\end{itemize}

\subsection{Algorithm accuracy}
\label{subsec:performance}

\begin{figure}[t]
  \centering
  \includegraphics{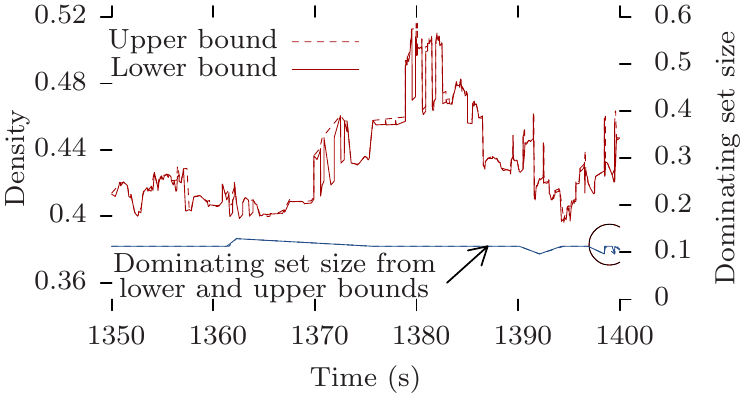}
  \caption{The upper and lower bounds give nearly identical values of density and dominating set size. Example taken from the Rollernet trace with $\tau$=5s and $\delta$=1min.\label{fig:error}}
\end{figure}

The approximation algorithm detailed in Section~\ref{sec:algorithm} is \textit{extremely accurate}. Indeed, thanks to Theorem~\ref{th:sampling} and Definition~\ref{def:approximate_composition}, the approximation calculates the \emph{exact} value for all times $t \in \eta\N$. Furthermore, the upper and lower bounds are nearly identical during epochs ($t \notin \eta\N$). Fig.~\ref{fig:error} plots the density over time for both the upper and lower bounds of the Rollernet reachability graph for $\tau=5$s and $\delta=1$min. The plot has been zoomed in to show only a small span of the $y$-axis for 30 seconds. The upper and lower bounds on density are nearly equal at all times. So are the values of the dominating set size computed from the upper and lower approximations that only disagree in the circled area on Fig.~\ref{fig:error}. These observations hold for all the reachability graphs computed in this paper. Indeed, looking at the values of average density and average dominating set of the over 5,000 pairs of upper/lower bound TVGs calculated in this paper, the maximum disagreement is $8.10^{-3}$ for the former and $4.10^{-2}$ for the latter. In practice, the difference is smaller than the width of the line in plots. Therefore, in the rest of this paper, we only plot the values based on the lower approximation.

\subsection{Revealing temporal structural properties}
\label{subsec:dynamic}

\begin{figure}[t]
  \centering
  \subfloat[$\delta=20$min\label{subfig:stanford_snap_20min}]{\includegraphics[width=0.45\columnwidth]{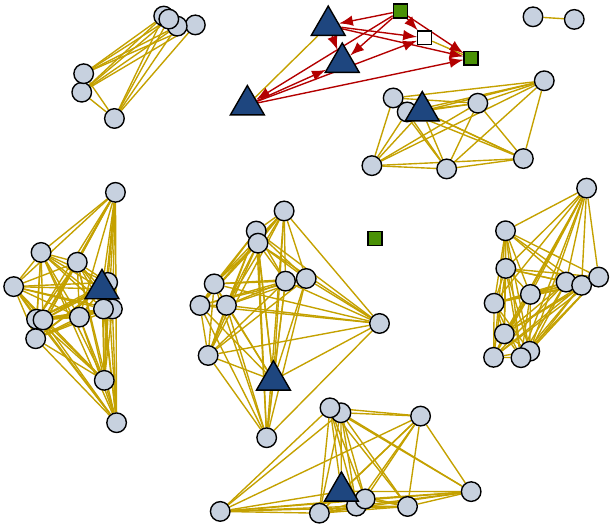}} \quad
  \subfloat[$\delta=40$min\label{subfig:stanford_snap_40min}]{\includegraphics[width=0.45\columnwidth]{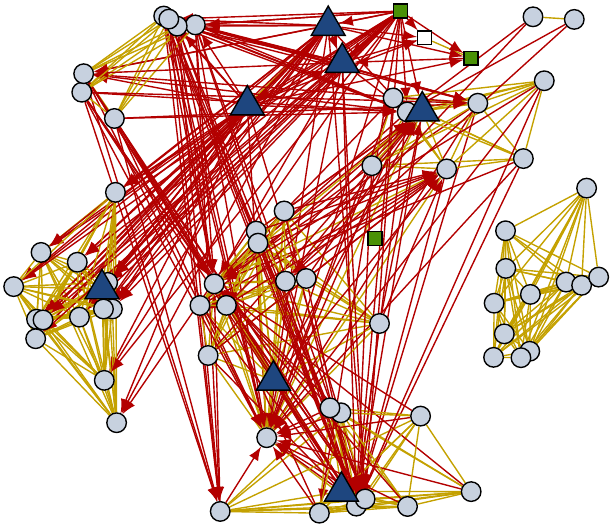}}
  \caption{Subset of Stanford's 20 and 40-minutes reachability graphs during morning classes ($\tau=20s$). Dark blue triangles are teachers; circles are students. The dark red arrows represent asymmetric arcs. The classroom structure is clearly visible.\label{fig:stanford_snap}}
\end{figure}

\begin{figure}[t]
  \centering
  \subfloat[$\delta=10$s\label{subfig:rollernet_snap_10s}]{\includegraphics[width=0.9\columnwidth]{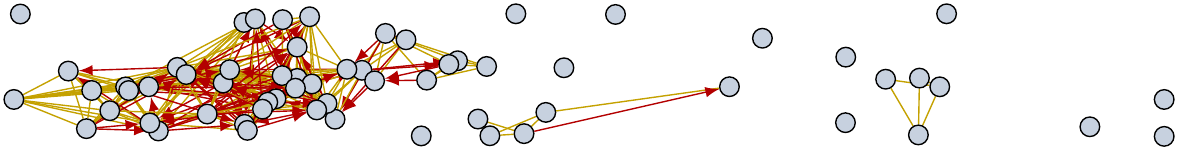}} \\
  \subfloat[$\delta=60$s\label{subfig:rollernet_snap_60s}]{\includegraphics[width=0.9\columnwidth]{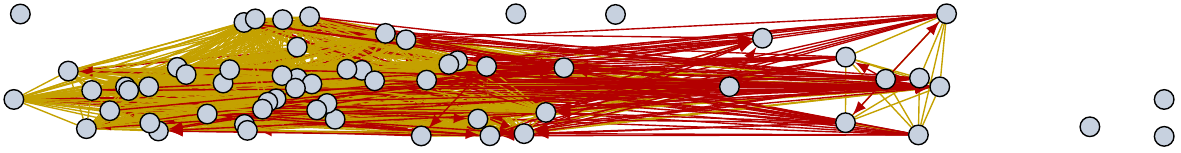}}
  \caption{Snapshots of rollernet's 10 and 60-seconds reachability graphs after 20 minutes ($\tau=5s$). The rollerblading tour is moving from left to right. The dark red arrows represent asymmetric arcs. The asymmetry is caused by the acceleration phase in the accordion phenomenon~\protect\cite{Tournoux2011}.  \label{fig:rollernet_snap}}
\end{figure}

\begin{figure*}[t]
  \centering
  \subfloat[Rollernet ($\tau=5$s)\label{subfig:rollernet_dynamic}]{\includegraphics{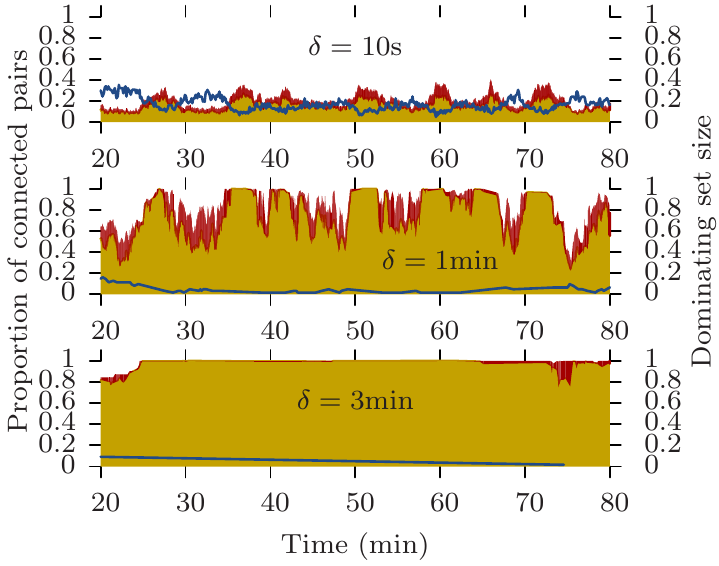}} \qquad
  \subfloat[Stanford ($\tau=1$s)\label{subfig:stanford_dynamic}]{\includegraphics{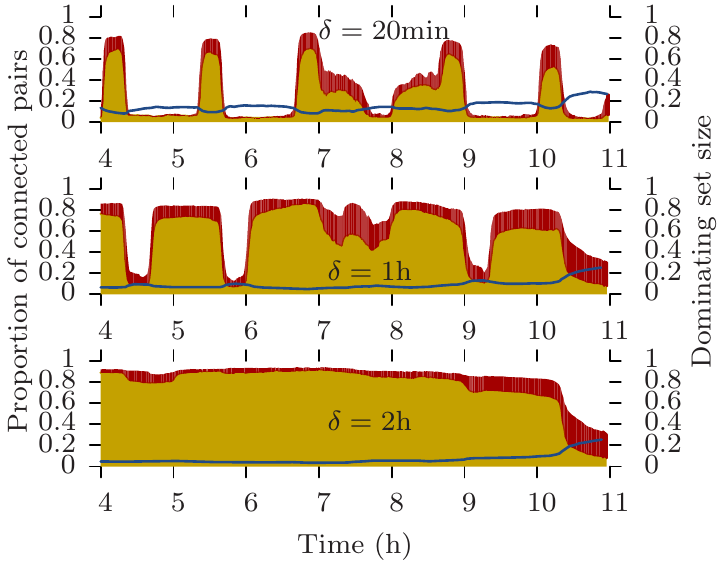}}
  \caption{Proportion of connected pairs of vertices over time for two real-life datasets. Connected pairs are divided into symmetric pairs (yellow) and asymmetric pairs (dark red).\label{fig:dynamic}}
\end{figure*}

The dynamics of reachability graphs highlight temporal structural properties of the original connectivity trace that are not otherwise accessible. Fig.~\ref{fig:stanford_snap} shows snapshots of a subset of the reachability graph in the high school network captured in the Stanford dataset. When the delay is 20 minutes, the classroom structure of the trace is clearly visible. In the original contact trace, edges within a classroom are unstable and lead to merges and splits of small connected components. In the reachability graph, a classroom is a stable complete subgraph with one teacher. A group of interacting teachers are also visible at the top of Fig.~\ref{subfig:stanford_snap_20min}. When some teachers and students later change classrooms, they create strong asymmetries in higher-delay reachability graphs such as the 40-minute one depicted in Fig.~\ref{subfig:stanford_snap_40min}.

Fig.~\ref{fig:rollernet_snap} shows snapshots of the reachability graph captured in the Rollernet dataset. These correspond to an acceleration phase, where the head of the tour (right on Fig.~\ref{fig:rollernet_snap}) pulls ahead of the rest of the rollerbladers. Due to the accordion phenomenon, the tail of the tour does not react immediately~\cite{Tournoux2011}. This completely prevents short delay communications between these two groups ($\delta=10$s on Fig.~\ref{subfig:rollernet_snap_10s}). With a longer delay (e.g., $\delta=60$s on Fig.~\ref{subfig:rollernet_snap_60s}), \textit{backward} communications towards the rear of the tour become possible (in particular through organizers who stop on the side of the road and let the tour pass them), thereby creating strong asymmetry in the reachability graph.

A more systematic study of dynamic properties over time is shown on Fig.~\ref{fig:dynamic}. It plots both the proportion of connected pairs of vertices (left axis) and the size of the dominating set (right axis). The pairs of vertices are further divided into symmetric pairs (yellow) and asymmetric pairs (red). At a given time, if the red histogram reaches 1 then a journey exists in at least one direction between all pairs of nodes. If the yellow histogram reaches 1, then a journey exist in both direction between all pairs of nodes.

For a 1-minute delay, the Rollernet reachability graph alternates, sometimes very rapidly, between fully connected states and highly asymmetric partially connected states (Fig.~\ref{subfig:rollernet_dynamic}). Therefore, any opportunistic communication system aiming for latencies under a minute will be strongly impacted by the accordion phenomenon. However, if a communication system can tolerate up to three minute delays, then it should be possible to smooth out the dynamic mobility. Indeed, the 3-minute reachability graph is fully connected for the entire duration of the trace, and the size of its dominating set is almost always equal to one.

The Stanford trace alternates static phases in classroom (the valleys on Fig.~\ref{subfig:stanford_dynamic}) and dynamic phases moving between classrooms or around the food court (the peaks on Fig~\ref{subfig:stanford_dynamic}). The progressive ``shrinking'' of the valleys illustrates how reachability graphs ``grow backwards'' with increasing delays. Indeed, if there exist a journey from $a$ to $b$ within 20 minutes at time $t$, then there exists a journey from $a$ to $b$ within 40 minutes at time $t$ minus $20$ minutes. Of course, after the increase in delays exceeds the width of the valleys, the reachability graph eventually reaches it maximum density. Note that in this case, an incompressible amount of asymmetry subsists throughout the trace.

Due to space constraints, we have only shown, in this section, results for specific values of $\tau$ but, as a general rule, the smaller the value of $\tau$, the faster the TRG becomes a clique. When analysing a network with a specific application in mind, the value of $\tau$ can be set to a realistic value (e.g. message size over bit-rate), and the plots in this section give an immediate visual understanding of the communication possibilites. The next section examines quantitatively the importance the delay ($\delta$) and the edge traversal time ($\tau$) parameters.

\subsection{Bounds on communication capabilities}
\label{subsec:communication_bounds}

\begin{figure*}[t]
  \centering
  \subfloat[Rollernet\label{subfig:rollernet_density}]{\includegraphics{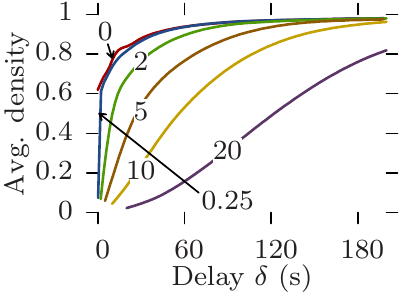}} \quad
  \subfloat[Stanford\label{subfig:stanford_density}]{\includegraphics{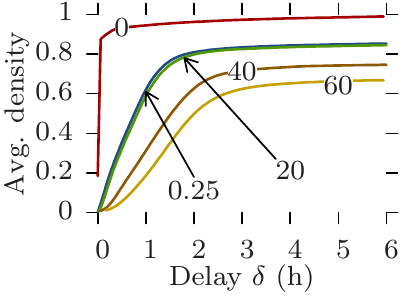}} \quad
  \subfloat[Random Waypoint\label{subfig:rwp_density}]{\includegraphics{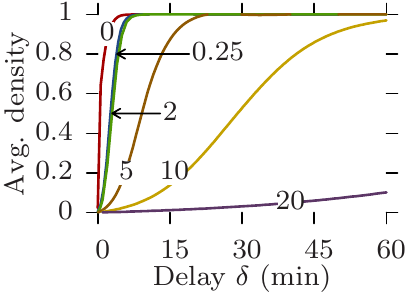}} \quad
  \subfloat[Community\label{subfig:community_density}]{\includegraphics{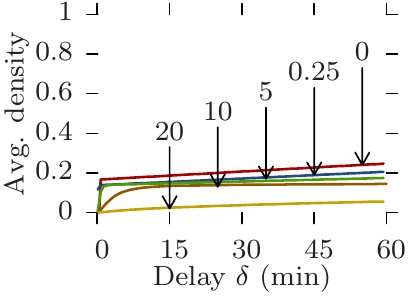}}
  \caption{Average density vs. maximum delay $\delta$ for different edge traversal times $\tau$ (in seconds).\label{fig:density}}
\end{figure*}

\begin{figure*}[t]
  \centering
  \subfloat[Rollernet\label{subfig:rollernet_ds_size}]{\includegraphics{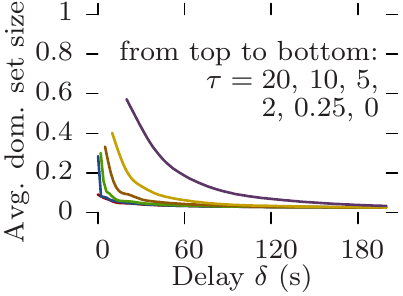}} \quad
  \subfloat[Stanford\label{subfig:stanford_ds_size}]{\includegraphics{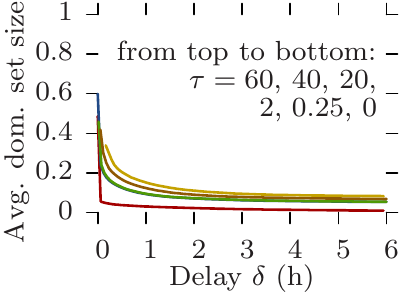}} \quad
  \subfloat[Random Waypoint\label{subfig:rwp_ds_size}]{\includegraphics{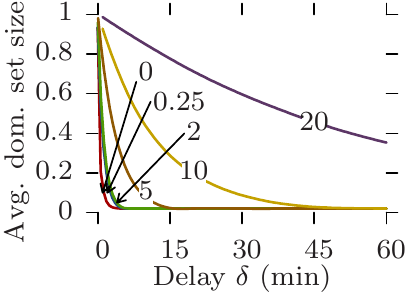}} \quad
  \subfloat[Community\label{subfig:community_ds_size}]{\includegraphics{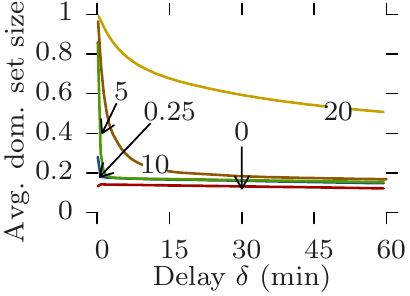}}
  \caption{Average dominating set size vs. maximum delay $\delta$ for different values of $\tau$ (in seconds).\label{fig:ds_size}}
\end{figure*}

Reachability graphs give straightforward bounds on communication capabilities. Indeed, the density at a given moment is exactly equal to the \textit{maximum delivery ratio expectancy} of a perfect opportunistic routing protocol whose delay-tolerance is equal to $\delta$ and whose message size over bit-rate ratio is equal to $\tau$. Furthermore, the size of the dominating set indicates the achievable \textit{offload ratio} in a scenario where the opportunistic network is assisting an infrastructure (e.g., 3G) for disseminating content to all nodes in the network~\cite{Han2011}.

Given real-world system requirements, i.e., wireless bit-rate estimates, messages sizes, target delivery ratio, and delay-tolerance, reachability graphs provide an immediate answer to the following question: can an opportunistic network support this service? If not, can it effectively supplement an infrastructure in an offloading scheme?

Fig.~\ref{fig:density} plots the average density against the delay tolerance for all datasets and for increasing edge traversal times, while Fig.~\ref{fig:ds_size} plots the average dominating set size. Rollernet and Random Waypoint share similar characteristics. Both are very sensitive to increasing edge traversal times. When these are close to one second, near 100\% density is achievable with a couple of minutes of delay tolerance (Figs.~\ref{subfig:rollernet_density}, \ref{subfig:rwp_density}, \ref{subfig:rollernet_ds_size}, and~\ref{subfig:rwp_ds_size}). In this case, they can support pure opportunistic communications. However, when edge traversal times are longer, tight delay constraints are impossible unless as part of an offloading scenario (e.g., $\tau=10$s and $\delta=60$s for Rollernet, Fig~\ref{subfig:rollernet_density}). For $\tau=20$s, Random Waypoint cannot even provide offloading for reasonable delays (Fig.~\ref{subfig:rwp_ds_size}).

Similarly, Stanford and Community share similar features. Regardless of the delay-tolerance and message size, no pure opportunistic routing protocol can provide anything near 100\% delivery ratio~\footnote{The $\tau=0$ results for Stanford are an artifact of the 20-second resolution.} (Figs.~\ref{subfig:stanford_density} and~\ref{subfig:community_density}). Despite this, they are both good offloading scenarios as the size of their dominating sets is consistently below 20\% of the total number of nodes, thereby offering potential offload ratios of around 80\%. Indeed, in order to disseminate content to the entire network, pushing one copy per classroom in the Stanford case, or one copy per community in the Community case, plus copies to single nodes is an obvious strategy. However, for larger values of $\tau$ (e.g., 20), Community is no longer able to offload content with reasonable delays (Fig.~\ref{subfig:community_ds_size}).

\subsection{Asymmetry}
\label{subsec:asymmetry}

\begin{figure*}
  \centering
  \subfloat[Rollernet\label{subfig:rollernet_asym}]{\includegraphics{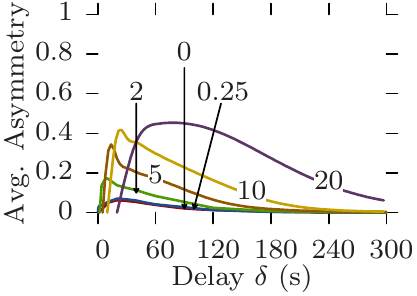}} \quad
  \subfloat[Stanford\label{subfig:stanford_asym}]{\includegraphics{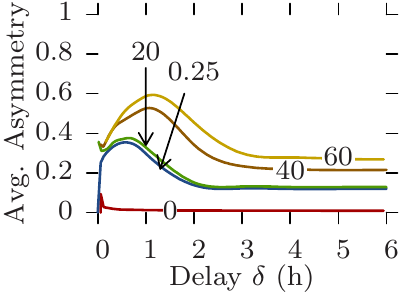}} \quad
  \subfloat[Random Waypoint\label{subfig:rwp_asym}]{\includegraphics{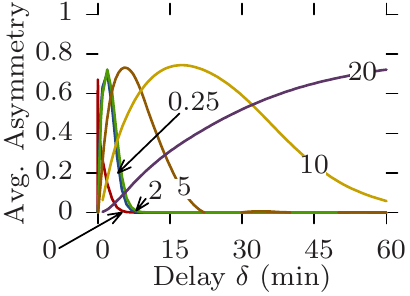}} \quad
  \subfloat[Community\label{subfig:community_asym}]{\includegraphics{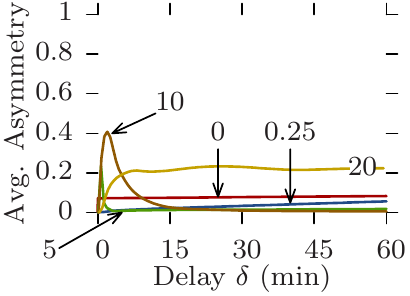}}
  \caption{Average asymmetry vs. maximum delay $\delta$ for different values of $\tau$ (in seconds).\label{fig:asym}}
\end{figure*}

Asymmetric communications are a fundamental aspect of opportunistic networks. Too strong focus on inter-contact times may lead to overlook asymmetry, but reachability graphs provide a natural way of studying and quantifying it. 

Fig.~\ref{fig:asym} plots for all datasets the average asymmetry against the maximum delay for increasing values of the edge traversal time. The asymmetry for the $\R_\tau$ graph directly derived from the original contact trace is always 0 as it only contains symmetric arcs (see Section~\ref{subsec:simple}). Asymmetry in Random Waypoint follows a regular pattern: a bell-shaped curve of constant max value, whose width and center increase with $\tau$ (Fig.~\ref{subfig:rwp_asym}). In this scenario, depending on the delay constraint, up to 80\% on average of connected pairs of vertices can only communicate in one direction. This is due to nodes traveling long straight distances creating asymmetric reachability to nodes they meet from nodes they had met earlier. As the maximum delay increases, return journeys using different intermediate nodes appear, the reachability becomes complete, and the asymmetry returns to 0. 

As previously, Rollernet and Random Waypoint show similar behavior. However, for Rollernet the maximum asymmetry also increases with $\tau$ (Fig.~\ref{subfig:rollernet_asym}). Indeed, as fewer edges are traversable, journeys become less likely between the front and the rear of the rollerblading tour. When such a journey is possible, it creates a longer-lasting asymmetry. In the Stanford trace, asymmetry follows the same increasing bell shape as in Rollernet, but never returns to 0 (Fig.~\ref{subfig:stanford_asym}). In fact, this incompressible \emph{minimum} asymmetry also increases with $\tau$.

\section{Related Work}
\label{sec:related}

Adaptations of traditional static graph distance metrics and algorithms to time-varying graphs have yielded many different concepts. For example, Orda \textit{et al.} propose a shortest path algorithm for TVGs based on different waiting policies (unrestricted, forbidden, and source waiting)~\cite{Orda1990}. Our work corresponds to the \textit{unrestricted} policy, in which a message may wait for an unlimited amount of time anywhere along its path through the TVG. Bui Xuan \textit{et al.} have proposed efficient algorithms for calculating shortest (in number of hops), fastest (in path traversal time), and foremost (i.e., earliest arrival) paths in TVGs~\cite{Xuan2003}. All these algorithms are designed to compute the shortest paths to all destination from a source and a fixed starting time. Our algorithms, in contrast, computes the reachibility graph by estimating the shortest paths for all possible starting times.

Several approaches to reachability in time-varying graphs exists. For strictly positive edge traversal times, a simple heuristic consists in dividing time into successive slots of length $\tau$ and keeping only edges that are persistently present during each slot~\cite{Daly2009}. This provides a good lower-bound approximation for small values of $\tau$ (i.e., less than $\eta$) whereas our approach can handle arbitrary edge traversal times. From a given starting time $t$, reachability among all pairs of nodes can be calculated by iterating over all edge UP/DOWN events~\cite{Tang2010}. This calculation can then be repeated for a sample of starting times~\cite{Holme2005}. This approach yields static reachability graphs for a discrete sequence of starting times, whereas the temporal reachability graphs defined in this paper calculate reachability in continuous time.

Chaintreau~\textit{et al.}, in their work on the diameter of opportunistic networks, calculate a \textit{Last Departure / Earliest Arrival} structure for each pair of nodes~\cite{chaintreau_diam}. This structure can tell for any pair of nodes $(A,B)$, at any time $t$, when is the earliest arrival for a message leaving $A$ for $B$ at time $t$. A reachability graph could be easily derived from these structures, but would unfortunately only cover the $\tau=0$ case. Our approach is more general since it can also handle non-zero edge traversal times.

\section{Conclusion and further work}
\label{sec:conclusion}

In this paper, we introduced the notion of \textit{temporal reachability graphs}. Given an edge traversal time and a maximum journey delay, temporal reachability graphs capture the temporal connectivity of the time-varying graphs they derive from. After formalizing the concept, we proved that regular reachability graphs, which encompass all experimental datasets, can be composed to compute reachability graphs of higher maximum delay with very high accuracy. Furthermore, we proposed a scalable highly-parallel streaming algorithm for their efficient computation. By applying this algorithm to synthetic and real-life contact traces, we showed how reachability graphs provide fresh new insights on temporal connectivity in time-varying graphs. In particular, they yield an immediate and intuitive picture of the communication capabilities and offloading potential of opportunistic networks. 

This work on reachability graphs will be pursued in several directions. Firstly, our results could be extended to the more general case where the edge traversal time $\tau$ is not constant but may take values among multiples of a TVG's resolution $\eta$. Secondly, as seen in this paper, reachability graphs seem to reveal community structures that are not immediately apparent in the contact traces. They could therefore lead to new approaches for the difficult problem of community detection in time-varying graphs. Finally, much work remains to be done on the statistical analysis and modeling of reachability graphs. For example, what are the correlations between symmetric and asymmetric arcs from a given node? How do degree and inter-arc time distributions evolve with edge traversal time and delay? If reachability graphs turn out to be easier to model than their underlying time-varying graphs, could they be used as a first step for realistic synthetic connectivity models?

\balance
\bibliographystyle{abbrv}

\pagebreak
\appendix

\section{Proofs}
\label{sec:proofs}

\subsection{Proof of Theorem~\ref{th:decomposition}}
\label{subsec:decomp_proof}

\begin{lemma}
\label{lemma:rd}
Let $\R_\delta$ be a reachability graph of a TVG $\G$ with edge traversal time $\tau$. Let $\J$ be a journey from $u$ to $v$ such that $departure(\J) \geq t$ and $arrival(\J) < t + \delta + \tau$. Then $\exists 0 \leq \epsilon < 1$ such that $(u,v) \in \R_{\delta+\epsilon\tau}(t)$.
\end{lemma}

\begin{proof}
If $arrival(\J) \leq t+\delta$, then we set $\epsilon=0$ and we have $(u,v) \in \R_\delta(t)$.

If $arrival(\J) > t+\delta$, let $0 \leq \epsilon < 1$ such that $arrival(\J) = t+\delta+\epsilon\tau$. In this case, $(u,v) \in \R_{\delta+\epsilon\tau}(t)$.
\end{proof}

\begin{lemma}
\label{lemma:rm}
Let $\R_\mu$ be a reachability graph of a TVG $\G$ with edge traversal time $\tau$. Let $\J$ be a journey from $u$ to $v$ such that $departure(\J) \geq t$ and $arrival(\J) \leq t + \mu$. Then $\mu \geq \tau$ and $\exists 0 \leq \epsilon < 1$ such that $(u,v) \in \R_{\mu-\epsilon\tau}(t+\epsilon\tau)$.
\end{lemma}

\begin{proof}
The temporal length of journey $\J$ is $\delta_{\J} \leq \mu$. Since it is a valid journey from $u$ to $v$, $\delta_{\J} \geq \tau$, hence $\mu \geq \tau$.

If $departure(\J) \geq t+\tau$, then $\delta_{\J} \leq \mu-\tau$, and $\forall 0 \leq \epsilon < 1$, $(u,v) \in \R_{\mu-\tau}(t+\tau) \subseteq \R_{\mu-\epsilon\tau}(t+\epsilon\tau)$.

If $departure(\J) < t+\tau$, then we set $0 \leq \epsilon < 1$ such that $departure(\J) = t+\epsilon\tau$. Then, since $arrival(\J) \leq t+\mu$, $\delta_{\J} \leq \mu-\epsilon\tau$ and $(u,v) \in \R_{\mu-\epsilon\tau}(t+\epsilon\tau)$.
\end{proof}

We can now prove Theorem~\ref{th:decomposition}.

\begin{proof}(Theorem~\ref{th:decomposition})
First, let us show that if an arc fits one of the decompositions then it belongs to $\R_{\delta+\mu}$. For all times $t$, let $(u,v)$ be an arc in $\left(\R_{\delta+\epsilon\tau} \otimes \R_{\mu-\epsilon\tau}\right)(t)$ ($0 \leq \epsilon < 1$). By definition, one of the following conditions holds.
\begin{enumerate}
\item $(u,v) \in \R_{\delta+\epsilon\tau}(t)$. Therefore $(u,v) \in \R_{\delta+\mu}(t)$ (Proposition~\ref{prop:growth}).
\item $(u,v) \in \R_{\mu-\epsilon\tau}(t+\delta+\epsilon\tau)$. In this case, there exists a journey $\J$ in $\G$ such that $departure(\J) \geq t+\delta+\epsilon\tau \geq t$ and $arrival(\J) \leq t + \delta +\epsilon\tau + \mu - \epsilon\tau \leq t + \delta + \mu$. Therefore, by definition, $(u,v) \in \R_{\delta+\mu}(t)$.
\item $\exists w \in V, s.t. (u,w) \in \R_{\delta+\epsilon \tau}(t) \text{ and } (w,v) \in \R_{\mu-\epsilon \tau}(t+\delta+\epsilon \tau)$. Let $\J$ and $\J'$ be the respective journeys from $u$ to $w$ and from $w$ to $v$. $\J \cup \J'$ is a journey from $u$ to $v$ such that $departure(\J \cup \J') = departure(\J) \geq t$ and $arrival(\J \cup \J) = arrival(\J') \leq t+\delta+\epsilon\tau+\mu-\epsilon\tau$, i.e. $arrival(\J \cup \J) \leq t + \delta + \mu$. Therefore $(u,v) \in \R_{\delta+\mu}(t)$.
\end{enumerate}

Conversely, let us show that for any time $t$ and any arc $(u,v) \in \R_{\delta+\mu}(t)$, there exists $0 \leq \epsilon < 1$, such that $(u,v) \in \left(\R_{\delta+\epsilon\tau} \otimes \R_{\mu-\epsilon\tau}\right)(t)$. Let $\J = \left\{(e_1,t_1), \dotsc, (e_k,t_k)\right\}$ be the journey from $u$ to $v$ such that $departure(\J) \geq t$ and $arrival(\J) \leq t+\delta+\mu$.  

Let $i = \max\left\{ 1 \leq j \leq k | t_j < t+\delta \right\}$.  There are three possible situations.
\begin{enumerate}
\item $i$ is not defined, i.e., $departure(\J) \geq t+\delta$. Since $arrival(\J) \leq t+\delta+\mu$, $\exists 0 \leq \epsilon <1$, such that $(u,v) \in \R_{\mu-\epsilon\tau}(t+\delta+\epsilon\tau)$ (Lemma~\ref{lemma:rm}).
\item $t_i=t_k$. Then $\J$ is a path from $u$ to $v$ such that $departure(\J) \geq t$ and $arrival(\J) <t+\delta+\tau$. Therefore $\exists 0 \leq \epsilon <1$, such that $(u,v) \in \R_{\delta+\epsilon\tau}(t)$ (Lemma~\ref{lemma:rd}).
\item $t_1 \leq t_i < t_{i+1} \leq t_k$. In this case, we can divide the journey $\J$ from $u$ to $v$ at time $t$ into the sub-journeys $\J_1 = \left\{ (e_1, t_1), \dotsc, (e_i,t_i)\right\}$ and $\J_2 = \left\{ (e_{i+1}, t_{i+1}), \dotsc, (e_k,t_k)\right\}$. We set $w = to(e_i) = from(e_{i+1})$. We have $arrival(\J_1) = t_i+\tau < t+\delta+\tau$. Therefore, $\exists 0 \leq \epsilon < 1$ such that $(u,w) \in \R_{\delta+\epsilon\tau}(t)$ (Lemma~\ref{lemma:rd}). Since $\J$ is a valid journey in $\G$, $t_{i+1} \geq t_i+\tau \geq t+\delta+\epsilon\tau$. Therefore $departure(\J_2) \geq t+\delta+\epsilon\tau$, and $(w,v) \in \R_{\mu-\epsilon\tau}(t+\delta+\epsilon\tau)$.
\end{enumerate}
In all three cases, $(u,v) \in \left(\R_{\delta+\epsilon\tau} \otimes \R_{\mu-\epsilon\tau}\right)(t)$.
\end{proof}

\subsection{Proof of Theorem~\ref{th:regular_reachability}}
\label{subsec:proof_regular_reachability}

\begin{lemma}[One-hop journeys]
\label{lemma:one_hop}
Let $\G$ be an $\eta$-regular TVG whose edge traversal time is $\tau \in \eta\N^*$. For $k \in \N$, if a one-hop journey $\{ (e,t) \}$ exists in $\G$ such that $ k\eta < t < (k+1)\eta$, then all one-hop journeys $\{ (e,t) \}$ with $ k\eta \leq t \leq (k+1)\eta$ also exist in $\G$.
\end{lemma}

\begin{proof}
Let $\tau = n\eta$ with $n \in \N^*$. If $\J = \{ (e,t) \}$ is a valid journey in $t$, then by definition for all $ t \leq \zeta < t+\tau$, $e \in \G(\zeta)$. Since $k\eta < t < (k+1)\eta$ and $\G$ is $\eta$-regular, then for all $k\eta \leq \zeta < (k+n)\eta, e \in \G(\zeta)$ (Definition~\ref{def:regular}). Furthermore, since $arrival(\J) = t+\tau > (k+n)\eta$, there exists $0 \leq \epsilon < 1$ such that $(k+n)\eta < (k+n+\epsilon)\eta < (k+n+1)\eta$. In this case, $e \in \G\left((k+n+\epsilon)\eta\right)$. Therefore, since $\G$ is $\eta$-regular, for all $(k+n)\eta \leq \zeta < (k+n+1)\eta$, $e \in \G(\zeta)$ (Definition~\ref{def:regular}). Finally let $\hat{\J} = \{ (e, \hat{t}) \}$ be a one-hop journey in $\G$ with $k\eta \leq \hat{t} < (k+1)\eta$. We have $departure(\hat{\J}) \geq k\eta$ and $arrival(\hat{\J}) \leq (k+1+n)\eta$. Since for all $k\eta \leq \zeta < (k+1+n)\eta$, $e \in \G(\zeta)$, $\hat{\J}$ is a valid journey in $\G$.
\end{proof}

\begin{lemma}[Epoch inclusion]
\label{lemma:epoch_inclusion}
Let $\G$ be an $\eta$-regular TVG whose edge traversal time is $\tau \in \eta\N^*$. For $\delta \in \eta\N$, let $\R_\delta$ be a reachability graph of $\G$. For $k \in \N$, let times $t_a$ and $t_b$ be such that $ k\eta < t_a < (k+1)\eta$ and $k\eta \leq t_b \leq (k+1)\eta$. Then $\R_\delta(t_a) \subseteq \R_\delta(t_b)$.
\end{lemma}

\begin{proof}
For $\delta=0$, $\R_\delta$ is always empty and the lemma is trivially true. Hereafter we write $\delta = d \eta$, with $d\in\N^*$. If an arc $(u,v)$ is in $\R_\delta(t_a)$, then there exists a journey $\J = \left\{ (e_1,t_1), \dotsc, (e_k,t_k) \right\}$ from $u$ to $v$ such that $t_1 \geq t_a$ and $t_k +\tau \leq t_a+\delta$.

First we assume that $t_a < t_b$. 

Let $i = \min\left\{ j | t_j \geq t_b + (j-1)\tau \right\}$. 

If $i$ is not defined then $\forall j, t_a + (j-1)\tau \leq t_j < t_b+(j-1)\tau$. We consider the journey $\hat{\J} = \left\{ (e_j, \hat{t_j} \right\}_{1 \leq j \leq k}$, such that $\hat{t_j} = t_b+(j-1)\tau$. By construction, the departure times of all of its one-hop journeys verify $\hat{t_j}-\hat{t}_{j-1} \geq \tau$. Furthermore, since $k\eta + (j-1)n\eta < t_j < \hat{t_j} < (k+1)\eta + (j-1)n\eta$, all one-hop journeys $\{ (e_j, \hat{t_j} ) \}$ are valid (Lemma~\ref{lemma:one_hop}). Therefore $\hat{\J}$ is a valid journey from $u$ to $v$ such that $departure(\hat{\J}) \geq t_b$ and $arrival(\hat{\J}) \leq t_b+\delta$, i.e., $(u,v) \in \R_\delta(t_b)$.

If $i$ is defined then there are two possibilities. If $i=1$ then $departure(\J) \geq t_b$ and $(u,v)$ is trivially in $\R_\delta(t_b)$. If $i>1$, then we can divide journey $\J$ into two subjourneys $\J_1 = \left\{ (e_j,t_j) \right\}_{0 \leq j < i}$ and $\J_2 = \left\{ (e_j,t_j) \right\}_{i \leq j \leq k}$. As above, we can transform $\J_1$ into a valid journey $\hat{\J_1} = \left\{ (e_j,\hat{t_j} \right\}_{0 \leq j < i}$ with $\hat{t_j} = t_b + (j-1)\tau$ such that $departure(\hat{J_1}) \geq t_b$ and $arrival(\hat{J_1}) \leq t_b+(i-1)\tau$. Since $departure(\J_2) \geq t_b + (i-1)\tau$, $\hat{\J_1} \cup \J_2$ is a valid journey from $u$ to $v$ that leaves after $t_b$ and arrives before $t_b+\delta$. Therefore $(u,v) \in \R_\delta(t_b)$.

Conversely, let us assume that $t_a > t_b$. Let $i = \max\left\{j | t_j + \tau \leq t_b + \delta - (k-j)\tau \right\}$.

If $i$ is not defined, then $\forall j, t_b+\delta+(k-j)\tau < t_j + \tau \leq t_a + \delta + (k-j)\tau$. We consider the journey $\hat{\J} = \left\{ (e_j, \hat{t_j} \right\}_{1 \leq j \leq k}$, such that $\hat{t_j} = t_a+\delta+(k-j-1)\tau$. By construction, the departure times of all of its one-hop journeys verify $\hat{t_j}-\hat{t}_{j-1} \geq \tau$. Furthermore, since $(k+d)\eta+ (k-j-1)n\eta < \hat{t_j} < t_j < (k+d+1)\eta + (k-j-1)n\eta$, all one-hop journeys $\{ (e_j, \hat{t_j} ) \}$ are valid (Lemma~\ref{lemma:one_hop}). Therefore $\hat{\J}$ is a valid journey from $u$ to $v$ such that $departure(\hat{\J}) \geq t_b$ and $arrival(\hat{\J}) \leq t_b+\delta$, i.e., $(u,v) \in \R_\delta(t_b)$.

If $i$ is defined then there are two possibilities. If $i=k$ then $arrival(\J) \leq t_b+\delta$ and $(u,v)$ is trivially in $\R_\delta(t_b)$. If $i<k$, then we can divide journey $\J$ into two subjourneys $\J_1 = \left\{ (e_j,t_j) \right\}_{0 \leq j \leq i}$ and $\J_2 = \left\{ (e_j,t_j) \right\}_{i < j \leq k}$. As above, we can transform $\J_2$ into a valid journey $\hat{\J_2} = \left\{ (e_j,\hat{t_j} \right\}_{i < j \leq k}$ with $\hat{t_j} = t_b + \delta+(k-j-1)\tau$ such that $departure(\hat{J_2}) \geq t_b+\delta+(k-i)\tau$ and $arrival(\hat{J_2}) \leq t_b+\delta$. Since $arrival(\J_1) \leq t_b+(k-i)\tau$, $\J_1 \cup \hat{\J_2}$ is a valid journey from $u$ to $v$ that leaves after $t_b$ and arrives before $t_b+\delta$. Therefore $(u,v) \in \R_\delta(t_b)$.
\end{proof}

We can now prove Theorem~\ref{th:regular_reachability}.

\begin{proof}
For $\delta=0$, $\R_\delta$ is always empty and the lemma is trivially true. Hereafter $\delta \in \eta\N^*$. For all $k\in\N$, and for all $k\eta \leq t_1 < (k+1)\eta$,  $\R_\delta(t_1) \subseteq \R_\delta(k\eta)$ and $\R_\delta(t_1) \subseteq \R_\delta\left((k+1)\eta)\right)$ (Lemma~\ref{lemma:epoch_inclusion}). For times $t_2$ such that $k\eta < t_2 < (k+1)\eta$, we have $\R_\delta(t_1) \subseteq \R_\delta(t_2)$ and $\R_\delta(t_2) \subseteq \R_\delta(t_1)$ (Lemma~\ref{lemma:epoch_inclusion}). Therefore $\R_\delta(t_1) = \R_\delta(t_2)$ and $\R_\delta$ is $\eta$-regular.
\end{proof}

\subsection{Proof of Theorem~\ref{th:sampling}}
\label{subsec:proof_sampling}

\begin{lemma}[Journeys in regular graphs]
\label{lemma:journeys_regular}
Let $\G$ be an $\eta$-regular TVG whose edge traversal time is $\tau \in \eta\N^*$. Let $\J$ be a journey in $\G$ from $u$ to $v$ such that $\exists (a,b) \in \N\times\N$, $departure(\J) \geq a\eta$ and $arrival(\J) < (b+1)\eta$. Then there exists another journey $\hat{\J}$ from $u$ to $v$ that verifies $departure(\J) \geq a\eta$ and $arrival(\J) \leq b\eta$.
\end{lemma}

\begin{proof}
We write $\tau = n\eta$ with $n \in \N^*$. If $b \leq a + n$, then $arrival(\J)-departure(\J) < \tau$ which is impossible. Therefore $b > a+n$.

$\J = \left\{ (e_1,t_1), \dotsc, (e_k,t_k) \right\}$ is a journey from $u$ to $v$. For all $i$, we note $c_i$ the integer such that $c_i\eta \leq t_i < (c_i+1)\eta$. Here $a=c_1$ and $b=c_k+n$. Indeed, we have $arrival(\J)=t_k+\tau<(c_k+n+1)\eta$. Then we define $\hat{\J} = \left\{ (e_1,c_1\eta), \dotsc, (e_k,c_k\eta) \right\}$.

Firstly, since $\J$ is a valid journey, for all $t_i \leq t < t_i + \tau, e_i \in \G(t)$. Because $\G$ is $\eta$-regular, this implies that for all $c_i\eta \leq t < c_i + \tau, e_i \in \G(t)$, and therefore all one-hop journeys $\{(e_i,c_i\eta)\}$ in $\hat{\J}$ are valid.

Secondly, for all $i$, $t_{i+1} - t_i \geq \tau$ leads to $(c_{i+1}-c_i)\eta+\eta > \tau$, i.e., $(c_{i+1}-c_i)+1 > n$. Since these are all integers, we get $c_{i+1}-c_i \geq n$ and all the one-hop journeys in $\hat{\J}$ may be taken successively and $\hat{\J}$ is therefore a valid journey from $u$ to $v$ such that $departure(\hat{\J})\geq a\eta$ and $arrival(\hat{\J}) \leq (c_k+n)\eta = b\eta$.
\end{proof}

We can now prove Theorem~\ref{th:sampling}.

\begin{proof}
We write $\tau = n\eta$ with $n \in \N^*$. At time $t=a\eta$, let $(u,v)$ be an arc in $\R_{(d+k)\eta}\otimes \R_{(m-k)\eta}$ with $0 \leq k < n$.  By setting $\epsilon=k/n$, $(u,v) \in \R_{(d+m)\eta}(a\eta)$ with $0 \leq \epsilon < 1$ (Theorem~\ref{th:decomposition}).

Conversely, let us consider an $(u,v) \in \R_{(d+m)\eta}(a\eta)$. Since $m\eta \geq \tau$,according to Theorem~\ref{th:decomposition}, there exists $0 \leq \epsilon < 1$ such that one of the following holds (we note $k$ the integer such that $k\eta \leq \epsilon n \eta < (k+1)\eta$).
\begin{itemize}
\item $(u,v) \in \R_{d\eta+\epsilon\tau}(a\eta)$. Therefore $(u,v) \in \R_{(d+k)\eta}(a\eta)$ (Lemma~\ref{lemma:journeys_regular}).
\item $(u,v) \in \R_{m\eta-\epsilon\tau}\left((a+d)\eta+\epsilon\tau\right)$. With the inclusion rule from Proposition~\ref{prop:growth}, $(u,v)$ belongs to\\$\R_{(m-k)\eta}\left((a+d)\eta+\epsilon n \eta\right)$. Furthermore, since $\R_{(m-k)\eta}$ is $\eta$-regular (Theorem~\ref{th:regular_reachability}), $(u,v) \in \R_{(m-k)\eta}\left((a+d+k)\eta\right)$.
\item $\exists w \in V, (u,w) \in \R_{d\eta+\epsilon \tau}(a\eta)$ and\\$(w,v) \in \R_{m\eta-\epsilon \tau}\left((a+d)\eta+\epsilon \tau\right)$. Following the same reasoning as for the first two conditions, we show that $\exists w \in V, (u,w) \in \R_{(d+k)\eta}(a\eta)$ and $(w,v) \in \R_{(m-k)\eta}\left((a+d+k)\eta\right)$.
\end{itemize}
In all three cases, there exists $0 \leq k < n$ such that $(u,v) \in \R_{(d+k)\eta}\otimes \R_{(m-k)\eta}$.
\end{proof}

\subsection{Proof of Proposition~\ref{prop:lower_epoch_approx}}
\label{subsec:proof_lower_epoch_approx}

\begin{proof}
For $d < 2n$, by definition $\Low_{d\eta} = \R_{d\eta}$. By induction, lets us suppose
the proposition true until a certain $l \geq 2n-1$, and show that it also holds for $l+1$.  Let $(u,v)$ be an arc in $\Low_{(l+1)\eta}(t)$. We write $a$ the integer such that $a\eta \leq t < (a+1)\eta$. If $t=a\eta$, then $\Low_{(l+1)\eta}(t) = \R_{(l+1)\eta}(t)$ (Definition~\ref{def:approximate_composition} and Theorem~\ref{th:sampling}) and $(u,v) \in \R_{(l+1)\eta}(t)$. Hereafter we assume that $t > a\eta$. Since $l+1 \geq 2n$, there exists $(d,m) \in N$ such that $n \leq d \leq l$, $n \leq m \leq l$, and $l+1 = d+m$. In this case, by definition, there exists $0 \leq k < n$ such that one of the following holds.
\begin{itemize}
\item $(u,v) \in \Low_{(d+k)\eta}\left((a+1)\eta\right)$. By setting $\epsilon = \frac{(a+1)\eta-t+k\eta}{n\eta}$, one can verify that $\exists 0 \leq \epsilon < 1$, such that $(u,v) \in \R_{d\eta+\epsilon \tau}(t)$.
\item $(u,v) \in \Low_{(m-k)\eta}\left((a+d+k)\eta\right)$. By setting $\epsilon = \frac{t-a\eta+k\eta}{n\eta}$, one can verify that $\exists 0 \leq \epsilon < 1$, such that $(u,v) \in \R_{m\eta-\epsilon \tau}(t+d\eta+\epsilon\tau)$.
\item $\exists w \in V, (u,w) \in \Low_{d\eta}(t)$ and $(w,v) \in \Low_{m\eta}(t+d\eta)$. By induction, $(u,w) \in \R_{d\eta}(t)$ and $(w,v) \in \R_{m\eta}(t+d\eta)$.
\end{itemize}
With Theorem~\ref{th:decomposition}, in all three cases, $(u,v) \in \R_{(l+1)\eta}$. By induction, this shows that $\forall d \in \N, d, \Low_{d\eta} \subseteq \R_{d\eta}$.
\end{proof}

\end{document}